\documentclass[11pt]{article}
\usepackage[latin1]{inputenc}   
\usepackage{amsmath}   
\usepackage{amsfonts}  
\usepackage{amssymb}  
\usepackage{latexsym}  
\usepackage{graphicx} 
\usepackage{epstopdf}
\usepackage[a4paper]{geometry}
\usepackage{color} 
\usepackage{subfigure}
\usepackage{url}
\usepackage{minitoc}
\usepackage{longtable}
\usepackage{algorithm}
\usepackage{algorithmic}
\usepackage{authblk}

\newtheorem{theorem}{Theorem}
\newtheorem{proposition}{Proposition}
\newtheorem{conjecture}{Conjecture}
\newtheorem{lemma}{Lemma}

\newtheorem{claim}{Claim}

\newcommand{\mc}[1]{\mathcal{#1}}
\newcommand{\mb}[1]{\mathbb{#1}}

\newcommand{\ccw}{counterclockwise }
\newcommand{\ccww}{counterclockwise}
\newcommand{\cw}{clockwise }
\newcommand{\cww}{clockwise}
\newcommand{\pss}{PS }
\newcommand{\ps}{Poulalhon and Schaeffer }
\newcommand{\npss}{Poulalhon and Schaeffer's }
\newcommand{\aps}{{\sc Algorithm PS} }
\newcommand{\apss}{{\sc Algorithm PS}}

\newenvironment{proof}{\noindent \emph{Proof.}\ }{\hfill
    $\Box$\vspace{1em}}

  \title{Encoding toroidal triangulations\thanks{This work was
      supported by the grant EGOS ANR-12-JS02-002-01 and the
    project-team GALOIS supported by LabEx PERSYVAL-Lab
    ANR-11-LABX-0025.}}

\author[1]{Vincent Despré
  \thanks{{vincent.despre@gipsa-lab.fr}}} \author[2]{Daniel
  Gon\c{c}alves\thanks{{daniel.goncalves@lirmm.fr}}}
\author[2]{ Benjamin
  L\'ev\^eque\thanks{{benjamin.leveque@lirmm.fr} }}

 \affil[1]{Université Joseph Fourier,
       GIPSA-Lab, Grenoble, France}

 \affil[2]{CNRS, Université de Montpellier, LIRMM, Montpellier, France}

\begin{document}
\maketitle

\begin{abstract}
  \ps introduced an elegant method to linearly encode a planar
  triangulation optimally. The method is based on performing a special
  depth-first search algorithm on a particular orientation of the
  triangulation: the minimal Schnyder wood. Recent progress toward
  generalizing Schnyder woods to higher genus enables us to generalize
  this method to the toroidal case.  In the plane, the method leads to
  a bijection between planar triangulations and some particular
  trees. For the torus we obtain a similar bijection but with
  particular unicellular maps (maps with only one face).
\end{abstract}


\section{Introduction}
\label{sec:introduction}

A closed curve on a surface is \emph{contractible} if it can be
continuously transformed into a single point. In this paper, we
consider graphs embedded on a surface such that they do not have
contractible cycles of size $1$ or $2$ (i.e. no contractible loops and
no multiple edges forming a contractible cycle). Note that this is a
weaker assumption, than the graph being \emph{simple}, i.e. not having
any cycle of size $1$ or $2$ (i.e. no loops and no multiple
edges). A graph embedded on a surface is called a \emph{map} on this
surface if all its faces are homeomorphic to open disks.  A map is a
triangulation if all its faces are triangles.  We denote by $n$ the
number of vertices, $m$ the number of edges and $f$ the number of
faces of a given map.

\ps introduced in~\cite{PS06} a method (called here PS method for
short) to linearly encode a planar triangulation with a binary word of
length
$\log_2\binom{4n}{n}\sim n\, \log_2(\frac{256}{27})\approx 3,2451\,n$
bits. This is asymptotically optimal since it matches the information
theory lower bound. The method is the following. Given a planar
triangulation $G$, it considers the minimal Schnyder wood of $G$ (that
is the orientation where all inner vertices have outdegree $3$ and
that contains no cycle oriented \cww). Then a special depth-first
search algorithm is applied by ``following'' ingoing edges and
``cutting'' outgoing ones. The algorithm outputs a rooted spanning tree
with exactly two leaves (also called stems) on each vertex from which
the original triangulation can be recovered in a straightforward
way. This tree can be encoded very efficiently.  A nice aspect of this
work, besides its interesting encoding properties, is that the method
gives a bijection between planar triangulations and a particular type
of plane trees.
 
Castelli Aleardi, Fusy and Lewiner~\cite{CFL10} adapt \pss method
to encode planar triangulations with boundaries. A consequence is that
a triangulation of any oriented surface can be encoded by cutting the
surface along non-contractible cycles and see the surface as a planar
map with boundaries. This method is a first attempt to generalize \pss
algorithm to higher genus. The obtained algorithm is asymptotically
optimal (in terms of number of bits) but it is not linear, nor
bijective.

The goal of this paper is to present a new generalization of
\pss algorithm to higher genus based on some strong structural
properties.  Applied on a  well chosen orientation of a
toroidal triangulation, what remains after the execution of the
algorithm is a rooted unicellular map (which corresponds to the
natural generalization of trees when going to higher genus) that can
be encoded optimally using $3,2451\,n$ bits.  Moreover, the algorithm
can be performed in linear time and leads to a new bijection between
toroidal triangulations and a particular type of unicellular maps.

The two main ingredients that make \pss algorithm work in an
orientation of a planar map are minimality and accessibility of the
orientation. \emph{Minimality} means that there is no \cw cycle.
\emph{Accessibility} means that there exists a root vertex such that
all the vertices have an oriented path directed toward the root
vertex.  Given $\alpha:V\to \mb{N}$, an orientation of $G$ is an
$\alpha$-orientation if for every vertex $v\in V$ its outdegree
$d^+(v)$ equals $\alpha(v)$. The existence and uniqueness of
  minimal orientations in the plane is given by the following result
of Felsner~\cite{Fel04} (related to older results of
Propp~\cite{Pro93} and Ossona de Mendez~\cite{Oss94}): the set of
$\alpha$-orientations of a given planar map carries a structure of
distributive lattice. This gives the existence and uniqueness of a
minimal $\alpha$-orientation as soon as an $\alpha$-orientation
exists.  Felsner's result enables several applications of \pss method
to other kind of planar maps, see~\cite{AP13,Ber07,DPS13}. In all
these cases the accessibility of the considered $\alpha$-orientations
is a consequence of the natural choice of $\alpha$, like in \npss
original work~\cite{PS06} where any $3$-orientation of the inner edges
of a planar triangulation is accessible for any choice of root vertex
on the outer face. (Note that the conventions may differs in the
literature: the role of outgoing and incoming edges are sometimes
exchanged and/or the role of \cw and \ccww.)

For higher genus, the minimality can be obtained by the following
generalization of Felsner's result.  The second author, Knauer and the
third author~\cite{GKL15} showed that on any oriented surface the set
of orientations of a given map having the same homology carries a
structure of distributive lattice. Note that $\alpha$ has been removed
here since it is captured by the homology (see
Section~\ref{sec:homology} for a brief introduction to homology).
Note also that this result is equivalent to an older result of
Propp~\cite{Pro93} where the lattice structure is described in the
dual setting.  Since this result is very general, there is hope to be
able to further generalize \pss method to surfaces. Note that a given
map on an oriented surface can have several $\alpha$-orientations (for
the same given $\alpha$) that are not homologous. So the set of
$\alpha$-orientations of a given map is now partitioned into
distributive lattices contrarily to the planar case where there is
only one lattice (and thus only one minimal element). In the case of
toroidal triangulations we manage to face this problem and maintain a
bijection by recent results on the structure of $3$-orientations of
toroidal triangulations. We identify a special lattice (and thus a
special minimal orientation) using the notion of Schnyder woods
generalized to the torus by the second and third author in~\cite{GL13}
(further generalized in~\cite{GKL15}, see also~\cite{LevHDR} for a
unified presentation).

The main issue while trying to extend \pss algorithm to higher genus
is the accessibility.  Accessibility toward the outer face is given
almost for free in the planar case because of Euler's formula that
sums to a strictly positive value. For an oriented surface of genus
$g\geq 1$ new difficulties occur. Already in genus $1$ (the torus),
even if the orientation is minimal and accessible \pss algorithm can
visit all the vertices but not all the angles of the map because of
the existence of non-contractible cycles.  We can show that the
special minimal orientation that we choose has the nice property that
this problem never occurs. In genus $g\geq 2$ things get even more
difficult with separating non-contractible cycles that make having
accessibility of the vertices already difficult to obtain.

Another problem is to recover  the original map after the
execution of the algorithm. If what remains after the execution of PS
method is an unicellular map then the map can be recovered with the
same simple rules as in the plane. Unfortunately for many minimal
orientations the algorithm leads to an unicellular embedded graph
covering all the edges but that is not a map (the only face is not a
disk) and it is not possible to directly recover the original map.
Here again, the choice of our special orientation ensures that this
never happens.

Finally the method presented here can be implemented in linear time.
Clearly the execution of \pss algorithm is linear but the difficulty
lies in providing the algorithm with the appropriate orientation in
input.  Computing the minimal Schnyder wood of a planar triangulation
can be done in linear time quite easily by using a so-called shelling
order (or canonical order, see~\cite{Kan96}). Other similar ad-hoc
linear algorithms can be found for other kinds of
$\alpha$-orientations of planar maps (see for
example~\cite[Chapter~3]{Fus07}). Such methods are not known in higher
genus. We solve this problems by first computing an orientation in our
special lattice and then go down in the lattice to find the minimal
orientation. All this can be performed in linear time.

Generalizing the method presented here to higher genus and other kind
of maps thus raises several challenging questions and we hope that the
present paper will lead to further generalizations of planar
bijections, coding, counting and sampling, to higher genus.

A brief introduction to homology and to the corresponding terminology
used in the paper is given in Section~\ref{sec:homology}.  In
Section~\ref{sec:schnyderwoods}, we present the definitions and
results we need concerning the generalization of Schnyder woods to the
toroidal case. In Section~\ref{sec:algo}, we introduce a reformulation
of \npss original algorithm that is applicable to any orientation of
any map on an oriented surface.  The main theorem of this paper is
proved in Section~\ref{sec:open}, that is, for a toroidal
triangulation given with an appropriate root and orientation, the
output of the algorithm is a toroidal unicellular map covering all the
edges of the triangulation. In Section~\ref{sec:close}, we show how
one can recover the original triangulation from the output.  This
output is then used in Section~\ref{sec:coding} to optimally encode a
toroidal triangulation. The linear time complexity of the method is
discussed in Section~\ref{sec:complexity}.  In
Section~\ref{sec:bijection} (resp. Section~\ref{sec:bij2}), we exhibit
a bijection between appropriately rooted toroidal triangulations and
rooted (resp. non-rooted) toroidal unicellular maps. To obtain the
non-rooted bijection, further structural results concerning the
particular Schnyder woods considered in this paper are given in
Section~\ref{sec:flip}.  Finally, a possible generalization to higher
genus is discussed in Section~\ref{sec:highergenus}.

\section{A bit of homology}
\label{sec:homology}

We need a bit of surface homology of general maps, which we discuss
now. The presentation is not standard but it is short and sufficient
to fit our needs. For a deeper introduction to homology we refer
to~\cite{Gib10}.

Consider a map $G$ with edge set $E$, on an orientable surface of
genus $g$, given with an arbitrary orientation of its edges. This
fixed arbitrary orientation is implicit and is used
to manipulate flows.  A \emph{flow} $\phi$ on $G$ is a vector in
$\mb{Z}^{E}$. For any $e\in E$, we denote by $\phi_e$ the
coordinate $e$ of $\phi$.

A \emph{walk} $W$ of $G$ is a sequence of edges such that consecutive
edges are incident.  A walk is \emph{closed} if the starting and
ending vertices are the same.  A walk has a \emph{characteristic flow}
$\phi(W)$ defined by:
$$\phi(W)_e:=\text{times }W\text{ traverses } e \text{ forward} - \text{times
}\\
W\text{ traverses } e \text{ backward}$$

This definition naturally extends to sets of walks.  From now on we
consider that a set of walks and its characteristic flow are the same
object.  We do similarly for oriented subgraphs as they can be seen as
sets of walks.

A \emph{facial walk} is a closed walk bounding a face.  Let $\mc{F}$
be the set of counterclockwise facial walks and let
$\mb{F}=<\phi(\mc{F})>$ the subgroup of $\mb{Z}^E$ generated by
$\mc{F}$.  Two flows $\phi, \phi'$ are said to be \emph{homologous} if
$\phi -\phi' \in \mb{F}$.  A flow $\phi$ is $0$-homologous if it is
homologous to the zero flow, i.e. $\phi \in \mb{F}$.

Let $\mc{W}$ be the set of \emph{closed} walks and let
$\mb{W}=<\phi(\mc{W})>$ the subgroup of $\mb{Z}^E$ generated by
$\mc{W}$.  The group $H(G)=\mb{W}/\mb{F}$ is the \emph{first homology
  group} of $G$.  Since $\dim(\mb{W})=m-n+1$ and $\dim(\mb{F})=f-1$,
Euler's Formula gives $\dim(H(G))=2g$.  So $H(G)\cong\mb{Z}^{2g}$ only
depends on the genus of the map.  A set $(B_1,\ldots,B_{2g})$ of
(closed) walks of $G$ is said to be a \emph{basis for the homology} if
$(\phi(B_1),\ldots,\phi(B_{2g}))$ is a basis of $H(G)$.

\section{Toroidal Schnyder woods}
\label{sec:schnyderwoods}

Schnyder~\cite{Sch89} introduced Schnyder woods for planar
triangulations using the following local property:

Given a map $G$, a vertex $v$ and an orientation and coloring of the
edges incident to $v$ with the colors $0$, $1$, $2$, we say that a
vertex $v$ satisfies the \emph{Schnyder property} if (see
Figure~\ref{fig:LSP}):

\begin{itemize}
\item Vertex $v$ has out-degree one in each color.
\item The edges $e_0(v)$, $e_1(v)$, $e_2(v)$ leaving $v$ in colors
  $0$, $1$, $2$, respectively, occur in counterclockwise order.
\item Each edge entering $v$ in color $i$ enters $v$ in the
  counterclockwise sector from $e_{i+1}(v)$ to $e_{i-1}(v)$ (where
  $i+1$ and $i-1$ are understood modulo $3$).
\end{itemize}

\begin{figure}[!h]
\center
\includegraphics[scale=0.5]{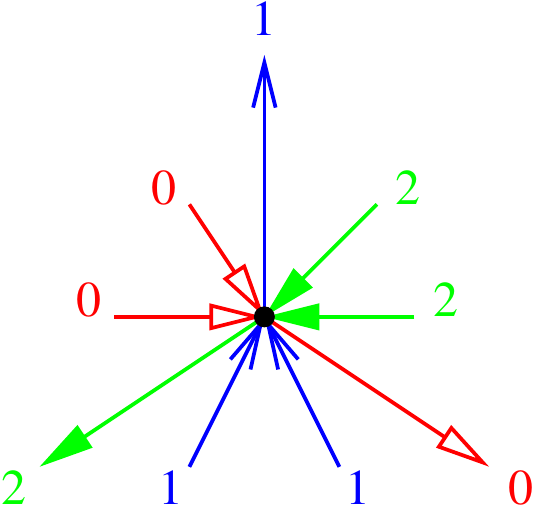}
\caption{The Schnyder property. The correspondence between red, blue,
  green and 0, 1, 2 and the arrow shapes used here serves as a
  convention for all figures in the paper.}
\label{fig:LSP}
\end{figure}

Given a planar triangulation $G$, a \emph{(planar) Schnyder wood} of
$G$ is an orientation and coloring of the inner edges of $G$ with the
colors $0$, $1$, $2$, where each inner vertex $v$ satisfies the
\emph{Schnyder property}.  In~\cite{GL13, GKL15} (see also the HDR
thesis of the third author~\cite{LevHDR}) a generalization of Schnyder
woods for higher genus has been proposed.  Since this paper deals with
triangulations of the torus only, we use a simplified version of the
definitions and results from~\cite{GL13, GKL15}.

The definition of Schnyder woods for toroidal triangulations is the
following.  Given a toroidal triangulation $G$, a \emph{(toroidal)
  Schnyder wood} of $G$ is an orientation and coloring of the edges of
$G$ with the colors $0$, $1$, $2$, where each vertex satisfies the
Schnyder property (see Figure~\ref{fig:tore-primal} for an
example). The three colors $0$, $1$, $2$ are completely symmetric in
the definition, thus we consider that two Schnyder woods that are
obtained one from the other by a (cyclic) permutation of the colors
are in fact the same object.  We consider that a Schnyder wood and its
underlying orientation are the same object since one can easily
recover a coloring of the edges in a greedy way (by choosing the color
of an edge arbitrarily and then satisfying the Schnyder property at
every vertex).

\begin{figure}[h!]
\center
\includegraphics[scale=0.5]{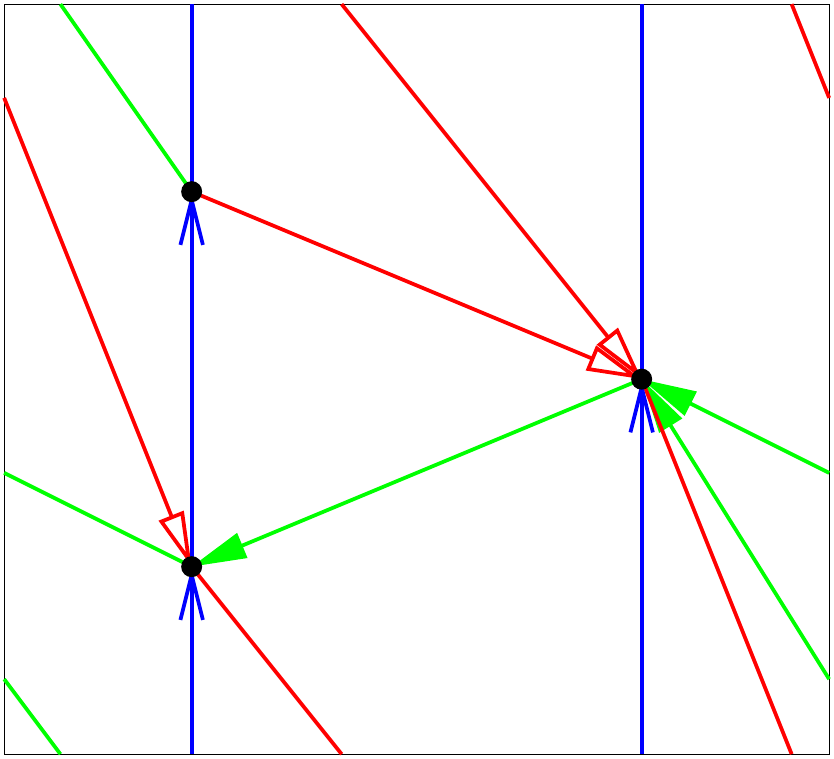}
\caption{A Schnyder wood of a toroidal triangulation (opposite sides
  are identified in order to form a torus).}
\label{fig:tore-primal}
\end{figure}

Note that the situation is quite different from the planar case.  In a
Schnyder wood of a toroidal triangulation, each vertex has exactly one
outgoing arc in each color, so there are monochromatic cycles
contrarily to the planar case (one can show that these monochromatic
cycles are not contractible). Moreover the graph induced by one color
is not necessarily connected.  Also, by a result of De Fraysseix and
Ossona de Mendez~\cite{FO01}, there is a bijection between
orientations of the internal edges of a planar triangulation where
every inner vertex has outdegree $3$ and Schnyder woods. Thus, in the
planar case, any orientation with the proper outdegrees corresponds to
a Schnyder wood. This is not true for toroidal triangulations since
there exists $3$-orientations that do not correspond to a Schnyder wood
(see Figure~\ref{fig:orientation}).

\begin{figure}[!h]
\center
\includegraphics[scale=0.5]{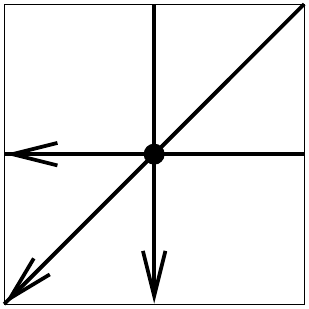} 
\hspace{5em}
\includegraphics[scale=0.5]{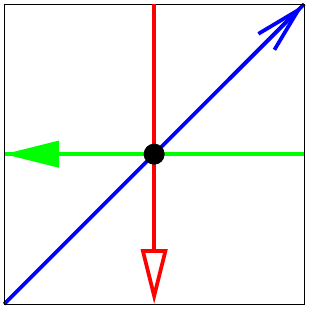}
\caption{Two different orientations of a toroidal triangulation. Only
  the one on the right corresponds to a Schnyder wood.}
\label{fig:orientation}
\end{figure}

A Schnyder wood of a toroidal triangulation is said to be
\emph{crossing}, if for each pair $i,j$ of different colors, there
exists a monochromatic cycle of color $i$ intersecting a monochromatic
cycle of color $j$. The existence of crossing Schnyder woods is proved
in~\cite[Theorem 1]{GL13} (note that in~\cite{GL13} the crossing
property is included in the definition of Schnyder woods,
see~\cite{LevHDR} for a unified presentation):

  \begin{theorem}[\cite{GL13}]
\label{th:crossing}
    A toroidal triangulation admits a crossing Schnyder wood.
  \end{theorem}

  Figure~\ref{fig:crossing} depicts two different Schnyder woods of
  the same graph where just the one on the left is crossing (on the
  right case the red and green monochromatic cycles do not intersect,
  we say that the Schnyder wood is ``half-crossing'' since blue crosses
  both green and red, see~\cite{GKL15,LevHDR} for a formal
  definition). Note that the Schnyder wood on the right is obtained
  from the one on the left by flipping a \cw triangle into a \ccw
  triangle.

 \begin{figure}[!h]
 \center
 \begin{tabular}{cc}
 \includegraphics[scale=0.4]{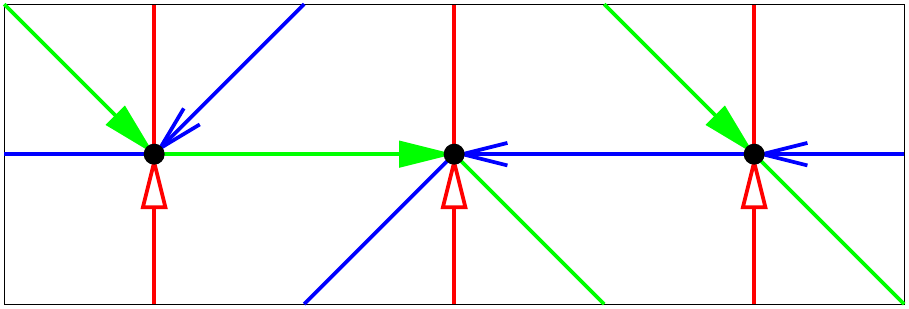} \ \ & \ \
 \includegraphics[scale=0.4]{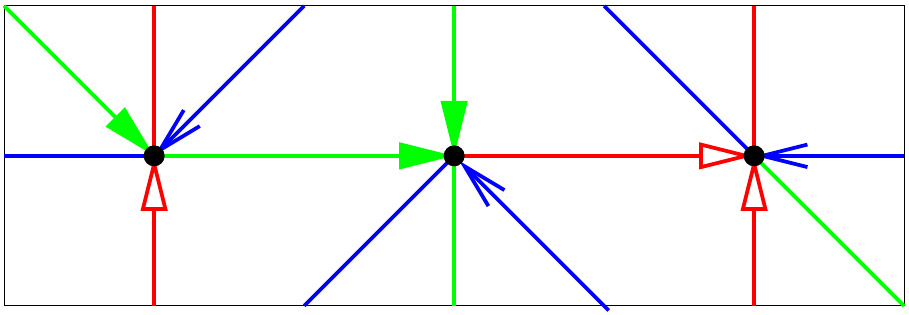} \\
 Crossing \ \ & \ \ Half-crossing \\
 \end{tabular}
 \caption{A crossing and an half-crossing Schnyder wood.}
 \label{fig:crossing}
 \end{figure}

 Consider a toroidal triangulation $G$ given with a crossing Schnyder
 wood. Let $D_0$ be the corresponding $3$-orientation of $G$.  Let
 $O(G)$ be the set of all the orientations of $G$ that are homologous
 to $D_0$. A consequence of~\cite[Theorem~5 and Corollary~2]{GKL15} is
 that all the crossing Schnyder woods of $G$ are homologous to each
 other.  So $O(G)$ contains all the crossing Schnyder woods of
 $G$. Thus the definition of $O(G)$ does not depend on the
 particular choice of $D_0$ and thus it is uniquely defined. A
 consequence of~\cite[Theorem~4 and Corollary~2]{GKL15} is that every
 orientation of $O(G)$ corresponds to a Schnyder wood. Thus we call
 the elements of $O(G)$ the \emph{homologous-to-crossing Schnyder
 woods} (or \emph{HTC Schnyder woods} for short).  Note that all the
 crossing Schnyder woods are HTC.

 Figure~\ref{fig:noncrossing} gives an example of an HTC Schnyder wood
 that is not crossing and a Schnyder woods that is not HTC. The
 example on the left is obtained from the crossing Schnyder wood of
 Figure~\ref{fig:crossing} by flipping two triangles (one to obtain
 the half-crossing Schnyder wood of Figure~\ref{fig:crossing} and then
 another one flipped from \ccw to \cww). Thus it is HTC since the
 difference with a crossing Schnyder wood is a 0-homologous oriented
 subgraph.  The example on the right of Figure~\ref{fig:noncrossing}
 is obtained from the crossing Schnyder wood of
 Figure~\ref{fig:crossing} by reversing the three vertical red
 monochromatic cycles. The union of these three cycles is not a
 0-homologous oriented subgraph, thus the resulting orientation is not
 HTC.

 \begin{figure}[!h]
 \center
 \begin{tabular}{cc}
 \includegraphics[scale=0.4]{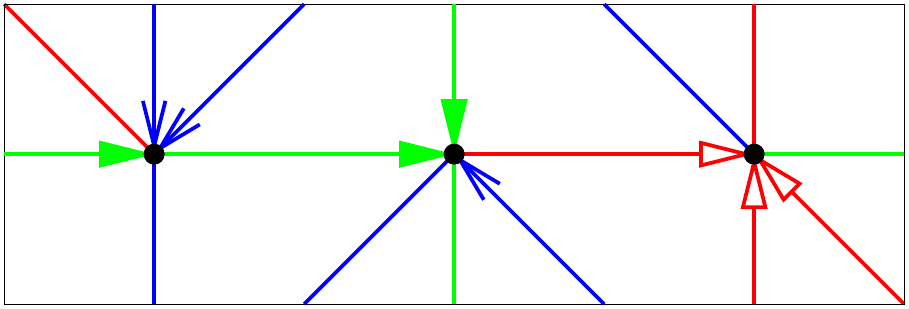} \ \ & \ \
 \includegraphics[scale=0.4]{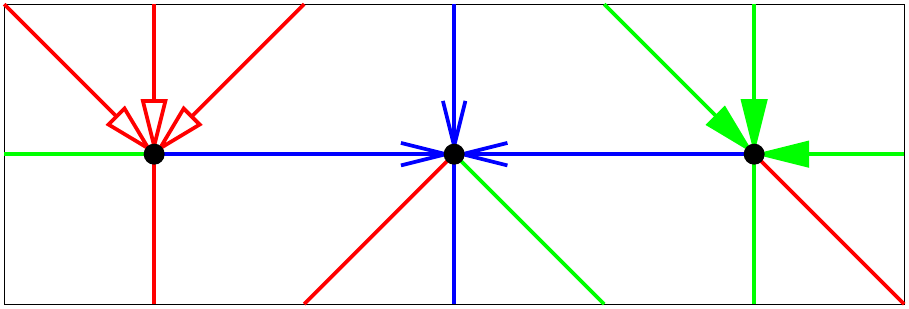} \\
 Non-crossing and HTC\ \ & \ \ Not HTC\\
 \end{tabular}
 \caption{Non-crossing Schnyder woods.}
 \label{fig:noncrossing}
 \end{figure}

 It is proved in~\cite{GKL15} that on any oriented surface the set of
 orientations of a given map having the same homology carries a
 structure of distributive lattice for a particular order defined
 below. Thus in particular the set of HTC Schnyder woods carries a
 structure of distributive lattice.
 
 Let us define an order on the orientations of $G$. For that purpose,
 choose an arbitrary face $f_0$ of $G$ and let $F_0$ be its
 counterclockwise facial walk (this choice of a particular face
 corresponds to the choice of the outer face in the planar case). Let
 $\mc{F}$ be the set of counterclockwise facial walks of $G$ and
 $\mc{F}'=\mc{F}\setminus F_0$.  We say that a $0$-homologous oriented
 subgraph $T$ of $G$ is \emph{counterclockwise} (resp.
 \emph{clockwise}) w.r.t.~$f_0$, if its characteristic flow can be
 written as a combination with positive (resp. negative) coefficients
 of characteristic flows of $\mc{F}'$,
 i.e. $\phi(T)=\sum_{F\in\mc{F}'}\lambda_F\phi(F)$, with
 $\lambda\in\mathbb{N}^{|\mc{F}'|}$
 (resp. $-\lambda\in\mathbb{N}^{|\mc{F}'|}$).  Given two orientations
 $D$ and $D'$ of $G$, let $D\setminus D'$ denote the subgraph of $D$
 induced by the edges that are not oriented as in $D'$.  We set
 $D\leq_{f_0} D'$ if and only if $D\setminus D'$ is counterclockwise.
 In~\cite[Theorem 7]{GKL15} the following is proved:

\begin{theorem}[\cite{GKL15}]
\label{th:lattice}
$(O(G),\leq_{f_0})$ is a distributive lattice.
\end{theorem}

Since $(O(G),\leq_{f_0})$ is a distributive lattice, it has a unique
minimal element.  The following lemma gives a property of
this minimum that is essential to apply \npss method.

\begin{lemma}
\label{lem:maximal}
The minimal element of $(O(G),\leq_{f_0})$ is the only HTC Schnyder
wood that contains no clockwise (non-empty) $0$-homologous oriented
subgraph w.r.t.~$f_0$.
\end{lemma}

\begin{proof}
  Let $D_{\min}$ be the minimal element of $(O(G),\leq_{f_0})$.
  Suppose by contradiction that $D_{\min}$ contains a clockwise
  non-empty $0$-homologous oriented subgraph $T$ w.r.t.~$f_0$.  The
  orientation of $G$ obtained from $D_{\min}$ by reversing all the
  edges of $T$ gives an orientation $D\in O(G)$ such that
  $T=D_{\min}\setminus D$.  Furthermore, by definition of
  $\leq_{f_0}$, we have $D\leq_{f_0} D_{\min}$, a contradiction to the
  minimality of $D_{\min}$.  So $D_{\min}$ contains no clockwise
  non-empty $0$-homologous oriented subgraph w.r.t.~$f_0$.
  
  We now show that this characterizes $D_{\min}$.  For any
  $D\in O(G)$, distinct from $D_{\min}$, we have
  $D_{\min}\leq_{f_0} D$. Thus $T=D\setminus D_{\min}$ is a non-empty
  clockwise $0$-homologous oriented subgraph of $D$.
\end{proof}

The crossing Schnyder wood of Figure~\ref{fig:tore-primal-min} is the
minimal HTC Schnyder wood for the choice of $f_0$ corresponding to the
shaded face. This example is used in the next sections to illustrate 
\npss method.

\begin{figure}[h!]
\center
\includegraphics[scale=0.5]{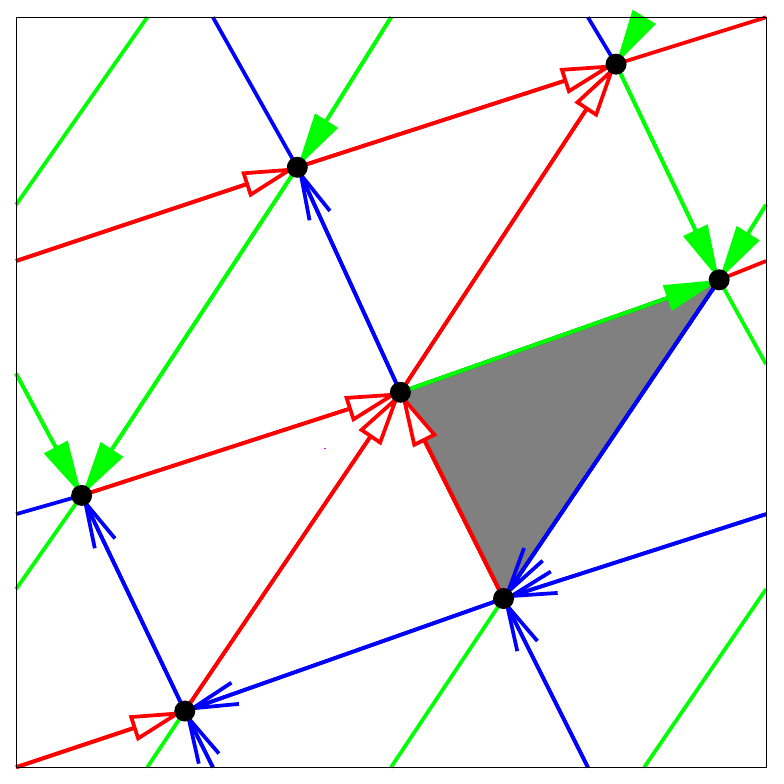}
\caption{The minimal HTC Schnyder wood of $K_7$ w.r.t.~the shaded
  face.}
\label{fig:tore-primal-min}
\end{figure}

The two HTC Schnyder woods of Figure~\ref{fig:crossing} are not
minimal (for any choice of special face $f_0$) since they contain
several triangles that are oriented clockwise. On the contrary, the
HTC Schnyder wood of Figure~\ref{fig:noncrossing} is minimal w.r.t to
its only face oriented clockwise. These examples shows that the
minimal HTC Schnyder wood is not always crossing.

We define the dual orientation $D^*$ of an orientation $D$ of $G$ as
an orientation of the edges of the dual map $G^*$ of $G$ satisfying
the following rule: the dual $e^*$ of an edge $e$ goes from the face
on the left of $e$ to the face on the right of $e$.  The following
lemma gives the key property of HTC Schnyder woods that we need in
this paper:

\begin{lemma}
\label{lem:incomingedges}
If $D$ is an orientation corresponding to an HTC Schnyder wood, then
the dual orientation $D^*$ contains no oriented non-contractible
cycle.
\end{lemma}

\begin{proof}
   We first prove the property for a crossing Schnyder
  wood and then show that it is stable by reversing a 0-homologous
  oriented subgraph.  Thus it is true for all HTC Schnyder woods.

  Consider a crossing Schnyder wood of $G$ by
  Theorem~\ref{th:crossing} and let $D_0$ be the corresponding
  orientation.  For $i\in\{0,1,2\}$, let $C_i$ be a monochromatic
  cycle of color $i$. By~\cite[Theorem 7]{GL13}, in a crossing
  Schnyder wood, the monochromatic cycles are not contractible and any
  two monochromatic cycles of different colors are not homologous and
  intersecting. Thus for $i\in\{0,1,2\}$, the two cycles $C_{i-1}$ and
  $C_{i+1}$ form a basis $B_i$ for the homology.  By the Schnyder
  property, cycle $C_{i-1}$ is crossing $C_i$ (maybe several time)
  from left to right.  Thus the homology of any closed curve can be
  expressed in at least one of the basis $B_i$ with only positive
  coefficients.

  Suppose now by contradiction that $D_0^*$ contains an oriented
  non-contractible cycle $C^*$. Let $i$ in $\{0,1,2\}$, such that
  $C^*$ is homologous to $\lambda_{i-1}C_{i-1}+\lambda_{i+1}C_{i+1}$
  with $\lambda_{i-1}>0$ and $\lambda_{i+1}\geq 0$. Then $C_{i+1}$ is
  crossing $C^*$ at least once from left to right, contradicting the
  fact that $C^*$ is an oriented cycle of $D_0^*$. So $D_0^*$ contains
  no oriented non-contractible cycle.

  Consider now an HTC Schnyder wood of $G$ and let $D$ be the
  corresponding orientation. Since $D$ and $D_0$ are both element of
  $O(G)$ they are homologous to each other. Let $T$ be the
  $0$-homologous oriented subgraph of $D$ such that
  $T=D\setminus D_0$. Thus $D_0$ is obtained from $D$ by reversing the
  edges of $T$.

  Suppose by contradiction that $D^*$ contains an oriented
  non-contractible cycle $C^*$. The oriented subgraph $T$ is
  $0$-homologous thus it intersects $C^*$ exactly the same number of
  time from right to left than from left to right. Since $C^*$ is
  oriented $T$ cannot intersect it from left to right. So $T$ does not
  intersect $C^*$ at all. Thus reversing $T$ to go from $D$ to $D_0$
  does not affect $C^*$. Thus $C^*$ is an oriented non-contractible
  cycle of $D_0^*$, a contradiction.
\end{proof}

For the non-HTC Schnyder wood of Figure~\ref{fig:noncrossing}, one can
see that there is an horizontal oriented non-contractible cycle in the
dual, so it does not satisfy the conclusion of
Lemma~\ref{lem:incomingedges}. Note that this property is not a
characterization of being HTC.  Figure~\ref{fig:gamma0glue} is a
Schnyder wood that is not HTC but satisfies the conclusion of
Lemma~\ref{lem:incomingedges} (we leave the reader check that this
Schnyder wood is not HTC, it will be easier after
Section~\ref{sec:bijection} and the definition of $\gamma$).

\begin{figure}[!h]
\center
\includegraphics[scale=0.4]{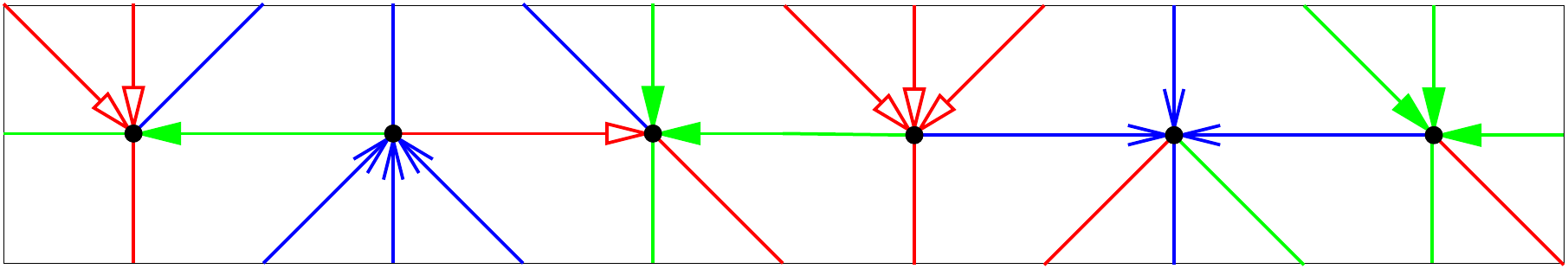}
\caption{A Schnyder wood that is not HTC but contains no oriented
  non-contractible cycle in the dual.}
\label{fig:gamma0glue}
\end{figure}

\section{\npss algorithm on oriented surfaces}
\label{sec:algo}

In this section we introduce a reformulation of \npss original
algorithm. This version is more general in order to be applicable to
any orientation of any map on an oriented surface. The execution
slightly differs from the original formulation, even on planar
triangulations. In~\cite{PS06}, the authors first delete some outer
edges of the triangulation before executing the algorithm. We do not
consider some edges to be special here since we want to apply the
algorithm on any surface but the core of the algorithm is the same. We
show general properties of the algorithm in this section before
considering toroidal triangulations in the forthcoming sections.

\

\noindent {\bf \sc Algorithm PS } 

{\sc Input} : An oriented map $G$, a root vertex $v_0$ and a root edge
$e_0$ incident to $v_0$.

{\sc Output} : An embedded graph $U$ with stems.
\begin{enumerate}
\item Let $v:=v_0$, $e:=e_0$, $U:=\emptyset$.
  \item 
Let $v'$ be the extremity of $e$ different from $v$.
\begin{itemize}
\item[\underline{Case 1} :] \emph{$e$ is non-marked and entering $v$.}
Add $e$ to $U$ and let $v:=v'$.
\item[\underline{Case 2} :] \emph{$e$ is non-marked and leaving $v$.}  Add a stem
  to $U$ incident to $v$ and corresponding to $e$.
\item[\underline{Case 3} :] \emph{$e$ is already marked and entering $v$.}
Do nothing.
\item[\underline{Case 4} :] \emph{$e$ is already marked and leaving $v$.}
Let $v:=v'$.
\end{itemize}
\item 
Mark $e$. 
\item Let  $e$ be the  next edge around $v$ in \ccw order after the current $e$.
\item While $(v,e)\neq (v_0,e_0)$ go back to 2.
\item Return $U$.
\end{enumerate}

We insist on the fact that the output of \aps is a graph embedded on
the same surface as the input map but that this embedded graph is not
necessarily a map (i.e some faces may not be homeomorphic to open
disks). In the following section we show that in our specific case the
output $U$ is an unicellular map.

Consider any oriented map $G$ on an oriented surface given with a root
vertex $v_0$ and a root edge $e_0$ incident to $v_0$. When \aps is
considering a couple $(v,e)$ we see this like it is considering the
angle at $v$ that is just before $e$ in clockwise order. The
particular choice of $v_0$ and $e_0$ is thus in fact a particular
choice of a root angle $a_0$ that automatically defines a root vertex
$v_0$, a root edge $e_0$, as well as a root face $f_0$. From now on we
 consider that the input of \aps is an oriented map plus a root
angle (without specifying the root vertex, face and edge).

The \emph{angle graph} of $G$, is the graph defined on the angles of
$G$ and where two angles are adjacent if and only if they are
consecutive around a vertex or around a face.  An execution of \aps
can be seen as a walk in the angle graph. Figure~\ref{fig:anglerule}
illustrates the behavior of the algorithm corresponding to Case~1
to~4. In each case, the algorithm is considering the angle in top left
position and depending on the marking of the edge and its orientation
the next angle that is considered is the one that is the end of the
magenta arc of the angle graph. The cyan edge of Case 1 represents the
edge that is added to $U$ by the algorithm. The stems of $U$ added in
Case~2 are not represented in cyan, in fact we will represent them
later by an edge in the dual. Indeed seeing the execution of \aps as a
walk in the angle graph enables us to show that \aps behaves exactly
the same in the primal or in the dual map (as explained later).

\begin{figure}[h!]
\center
\begin{tabular}[c]{ll}
\includegraphics[scale=0.5]{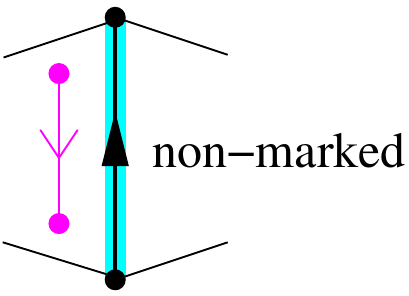} \ \ \ \ &
\includegraphics[scale=0.5]{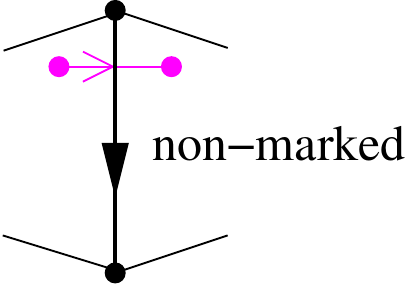} \\
\hspace{1em} Case 1 & \hspace{1em} Case 2 \\
\ &\  \\
\includegraphics[scale=0.5]{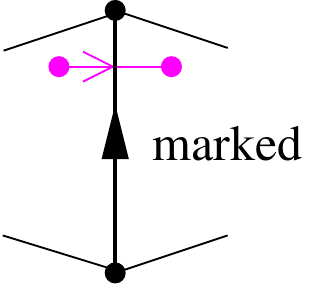} \ \ \ \ &
\includegraphics[scale=0.5]{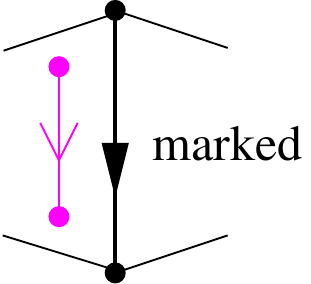} \\
\hspace{1em} Case 3 & \hspace{1em} Case 4
\end{tabular}
\caption{The four cases of \aps.}
\label{fig:anglerule}
\end{figure}

On Figure~\ref{fig:tore-example}, we give an example of an execution
of \aps on the orientation corresponding to the minimal HTC Schnyder
wood of $K_7$ of Figure~\ref{fig:tore-primal-min}.

\begin{figure}[h!]
\center
\begin{tabular}{cc}
\includegraphics[scale=0.5]{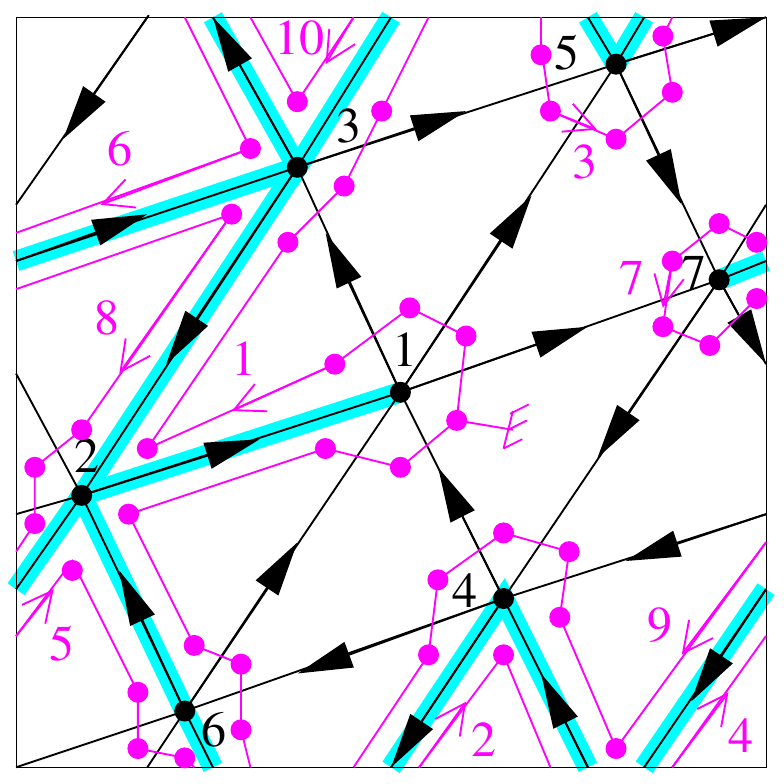} \ & \
\includegraphics[scale=0.5]{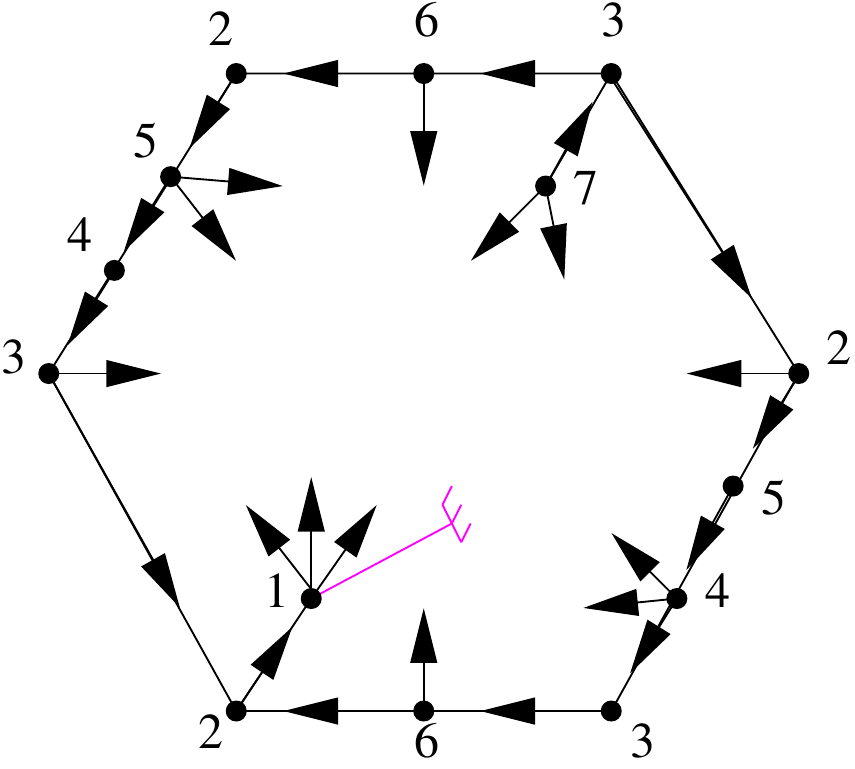}\\
Execution \ &\ Output \\
\end{tabular}
\caption{An execution of \aps on $K_7$ given with the orientation
  corresponding to the minimal HTC Schnyder wood of
  Figure~\ref{fig:tore-primal-min}. Vertices are numbered in
  black. The root angle is identified by a root symbol and chosen in
  the face for which the orientation is minimal (i.e. the shaded face
  of Figure~\ref{fig:tore-primal-min}). The magenta arrows and numbers
  are here to help the reader to follow the cycle in the angle graph.
  The output $U$ is a toroidal unicellular map, represented here as an
  hexagon where the opposite sides are identified.}
\label{fig:tore-example}
\end{figure}

Let $a$ be a particular angle of the map $G$. It is adjacent to four other angles
in the \emph{angle graph} (see Figure~\ref{fig:angledirected}). Let
$v,f$ be such that $a$ is an angle of vertex $v$ and face $f$. The
\emph{next-vertex} (resp. \emph{previous-vertex}) angle of $a$ is the
angle appearing just after (resp. before) $a$ in \ccw order around
$v$. Similarly, the \emph{next-face} (resp. \emph{previous-face})
angle of $a$ is the angle appearing just after (resp. before) $a$ in
\cw order around $f$. These definitions enable one to orient
consistently the edges of the angle graph like in
Figure~\ref{fig:angledirected} so that for every oriented edge
$(a,a')$, $a'$ is a next-vertex or next-face angle of $a$.

\begin{figure}[h!]
\center
\includegraphics[scale=0.5]{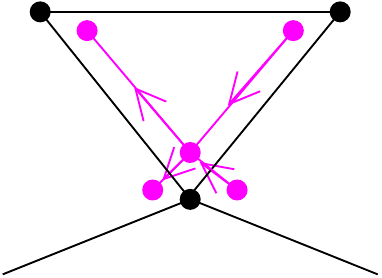} 
\caption{Orientation of the edges of the angle graph.}
\label{fig:angledirected}
\end{figure}

The different cases depicted in Figure~\ref{fig:anglerule} show that
an execution of \aps is just an oriented walk in the angle graph
(i.e. a walk that is following the orientation of the edges described in
Figure~\ref{fig:angledirected}).  The condition in the while loop
ensures that when the algorithm terminates, this walk is back to the
root angle. The following proposition shows that the algorithm actually
terminates:

\begin{proposition}
  \label{prop:terminates}
  Consider an oriented map $G$ on an oriented surface and a root angle
  $a_0$. The execution of \aps on $(G,a_0)$ terminates and corresponds
  to a cycle in the angle graph.
\end{proposition}
\begin{proof}
  We consider the oriented walk $W$ in the angle graph corresponding
  to the execution of \aps. Note that $W$ may be infinite. The walk
  $W$ starts with $a_0$, and if it is finite it ends with $a_0$ and
  contains no other occurrence of $a_0$ (otherwise the algorithm
  should have stopped earlier). Toward a contradiction, suppose that
  $W$ is not simple (i.e. some angles different from the root angle
  $a_0$ are repeated). Let $a\neq a_0$ be the first angle along $W$
  that is met for the second time.  Let $a_1,a_2$ be the angles
  appearing before the first and second occurrence of $a$ in $W$,
  respectively. Note that $a_1\neq a_2$ by the choice of $a$.

  If $a_1$ is the previous-vertex angle of $a$, then $a_2$ is the
  previous-face angle of $a$. When the algorithm considers $a_1$, none
  of $a$ and $a_2$ are already visited, thus edge $e$ is not
  marked. Since the execution then goes to $a$ after $a_1$, we are in
  Case 2 and the edge $e$ between $a$ and $a_1$ is oriented from $v$,
  where $v$ is the vertex incident to $a$. Afterward, when the
  algorithm reaches $a_2$, Case 3 applies and the algorithm cannot go
  to $a$, a contradiction. The case where $a_1$ is the previous-face
  angle of $a$ is similar.

  So $W$ is simple. Since the angle graph is finite, $W$ is finite. So
  the algorithm terminates, thus $W$ ends on the root angle and $W$ is
  a cycle.
\end{proof}

In the next section we see that in some particular cases the
cycle in the angle graph corresponding to the execution of \pss
algorithm (Proposition~\ref{prop:terminates}) can be shown to be
Hamiltonian like on Figure~\ref{fig:tore-example}.

By Proposition~\ref{prop:terminates}, an angle is considered at most
once by \apss. This implies that the angles around an edge can be
visited in different ways depicted on Figure~\ref{fig:anglerulesum}.
Consider an execution of \aps on $G$.  Let $C$ be the cycle formed in
the angle graph by Proposition~\ref{prop:terminates}.  Let $P$ be the
set of edges of the output $U$ (without the stems) and $Q$ be the set
of dual edges of edges of $G$ corresponding to stems of $U$. These
edges are represented on Figure~\ref{fig:anglerulesum} in cyan for $P$
and in yellow for $Q$.  They are considered with their orientation
(recall that the dual edge $e^*$ of an edge $e$ goes from the face on
the left of $e$ to the face on the right of $e$).  Note that $C$ does
not cross an edge of $P$ or $Q$, and moreover $P$ and $Q$ do not
intersect (i.e. an edge can be in $P$ or its dual in $Q$ but both
cases cannot happen).

\begin{figure}[h!]
\center
\begin{tabular}[c]{cc}
  \includegraphics[scale=0.5]{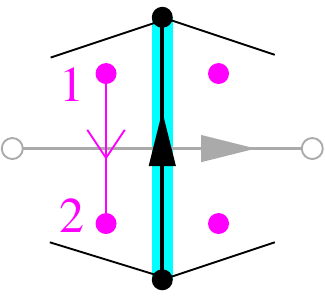} \ \ & \ \
\includegraphics[scale=0.5]{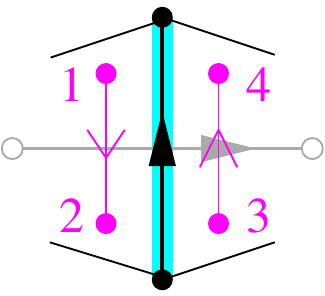} \\
 Case 1 alone \ \ & \ \ Case 1 and then 4 \\
\ \\
\includegraphics[scale=0.5]{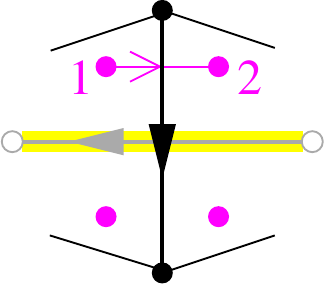} \ \ & \ \
\includegraphics[scale=0.5]{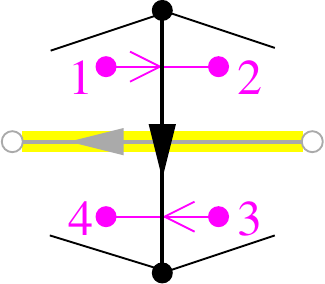} \\
 Case 2 alone \ \ & \ \ Case 2 and then 3 \\
\end{tabular} \ \ \ \
 \begin{tabular}[c]{cc}
   \includegraphics[scale=0.5]{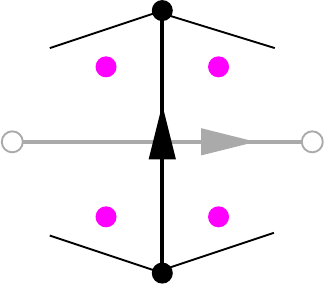} \\
 Not visited
 \end{tabular}
 \caption{The different cases of \aps seen in a dual way. The number
   of the angles gives the order in which the algorithm visits them
   (unvisited angles are not numbered). The edges of $P$ and $Q$
    are respectively cyan and yellow.}
\label{fig:anglerulesum}
\end{figure}

One can remark that the cases of Figure~\ref{fig:anglerulesum} are
dual of each other. One can see that \aps behaves exactly the same if
applied on the primal map or on the dual map. The only modifications
to make is to start the algorithm with the face $f_0$ as the root
vertex, the dual of edge $e_0$ as the root edge and to replace \ccw by
\cw at Line 4. Then the cycle $C$ formed in the angle graph is exactly
the same and the output is $Q$ with stems corresponding to $P$
(instead of $P$ with stems corresponding to $Q$). Note that this
duality is also illustrated by the fact that the minimality of the
orientation of $G$ w.r.t.~ the root face is nothing else than the
accessibility of the dual orientation toward the root face. Indeed, a
clockwise $0$-homologous oriented subgraph of $G$ w.r.t $f_0$
corresponds to a directed cut of the dual where all the edges are
oriented from the part containing $f_0$. The following lemma shows the
connectivity of $P$ and $Q$:

\begin{lemma}
\label{lem:PDconnected}
At each step of the algorithm, for every vertex $v$ appearing in an edge
of $P$ (resp. $Q$), there is an oriented path from $v$ to $v_0$ (resp.
$f_0$) consisting only of edges of $P$ (resp. $Q$). In particular $P$
and $Q$ are connected.
\end{lemma}
\begin{proof}
  If at a step a new vertex is reached then it correspond to Case 1
  and the corresponding edge is added in $P$ and oriented from the new
  vertex, so the property is satisfied by induction. As observed
  earlier the algorithm behaves similarly in the dual map.
\end{proof}

Let $\overline{C}$ be the set of angles of $G$ that are not in $C$.
Any edge of $G$ is bounded by exactly 4 angles. Since $C$ is a cycle,
the 4 angles around an edge are either all in $C$, all in
$\overline{C}$ or 2 in each set (see
Figure~\ref{fig:anglerulesum}). Moreover, if they are 2 in each set,
these sets are separated by an edge of $P$ or an edge of $Q$. Hence
the frontier between $C$ and $\overline{C}$ is a set of edges of $P$
and $Q$. Moreover this frontier is an union of oriented closed walks
of $P$ and of oriented closed walks of $Q$.  In the next section we
study this frontier in more details to show that $\overline{C}$ is
empty in the case considered there.

\section{From toroidal triangulations to unicellular maps}
\label{sec:open}

Let $G$ be a toroidal triangulation.  In order to choose appropriately
the root angle $a_0$, we have to consider separating triangles.  A
\emph{triangle} is a closed walk of size $3$ (it is not necessarily a
cycle since non-contractible loops are allowed and it is not
necessarily contractible).  A \emph{separating triangle} is a
contractible triangle that is different from a face of $G$. We say
that an angle is \emph{in the strict interior of a separating
  triangle} if it is in its contractible region and not incident to a
vertex of the triangle.  We choose as root angle $a_0$ any angle that
is not in the strict interior of a separating triangle.  One can
easily see that such an angle $a_0$ always exists. Indeed the
interiors of two contractible triangles are either disjoint or one is
included in the other. So, the angles that are incident to a
contractible triangle whose interior is maximal by inclusion satisfy
the property.

The goal of this section is to prove the following theorem (see
Figure~\ref{fig:tore-example} for an example):

\begin{theorem}
\label{th:uni}
Consider a toroidal triangulation $G$, a root angle $a_0$ that is not
in the strict interior of a separating triangle and the orientation of
the edges of $G$ corresponding to the minimal HTC Schnyder wood
w.r.t.~the root face $f_0$ containing $a_0$. Then the output $U$ of
\aps applied on $(G,a_0)$ is a toroidal unicellular map covering all
the edges of $G$.
\end{theorem}

Consider a toroidal triangulation $G$, a root angle $a_0$ that is not
in the strict interior of a separating triangle and the orientation of
the edges of $G$ corresponding to the minimal HTC Schnyder wood
w.r.t.~the root face $f_0$ containing $a_0$. Let $U$ be the output of
\aps applied on $(G,a_0)$.  We use the same notation as in previous
section: the cycle in the angle graph is $C$, the set of angles that
are not in $C$ is $\overline{C}$, the set of edges of $U$ is $P$, the
dual edges of stems of $U$ is $Q$.

\begin{lemma}
\label{lem:noD}
The frontier between $C$ and $\overline{C}$ contains no oriented
closed walk of $Q$.
\end{lemma}

\begin{proof}
  Suppose by contradiction that there exists such a walk $W$. Then
  along this walk, all the dual edges of $W$ are edges of $G$ oriented
  from the region containing $C$ toward $\overline{C}$ as one can see
  in Figure~\ref{fig:anglerulesum}.  If $W$ is non-contractible, then
  $W$ contains an oriented non-contractible cycle, a contradiction to
  Lemma~\ref{lem:incomingedges}. So $W$ is contractible.  So it
  contains an oriented contractible cycle $W'$, and then either $C$ is
  in the contractible region delimited by $W'$, or not. The two case
  are considered below:

  Suppose first that $C$ lies in the non-contractible region of $W'$.
  Then consider the plane map $G'$ obtained from $G$ by keeping only
  the vertices and edges that lie (strictly) in the contractible
  region delimited by $W'$. Let $n'$ be the number of vertices of
  $G'$. All the edges incident to $G'$ that are not in $G'$ are
  entering $G'$. So in $G'$ all the vertices have outdegree $3$ as we
  are considering $3$-orientations of $G$. Thus the number of edges of
  $G'$ is exactly $3n'$, contradicting the fact that the maximal
  number of edges of planar map on $n$ vertices is $3n-6$ by Euler's
  formula.

  Suppose now that $C$ lies in the contractible region of $W'$.  All
  the dual edges of $W'$ are edges of $G$ oriented from its
  contractible region toward its exterior. Consider the graph
  $G_{out}$ obtained from $G$ by removing all the edges that are cut
  by $W'$ and all the vertices and edges that lie in the contractible
  region of $W'$. As $G$ is a map, the face of $G_{out}$ containing
  $W'$ is homeomorphic to an open disk. Let $F$ be its facial walk (in
  $G_{out}$) and let $k$ be the length of $F$.  We consider the map
  obtained from the facial walk $F$ by putting back the vertices and
  edges that lied inside. We transform this map into a plane map $G'$
  by duplicating the vertices and edges appearing several times in
  $F$, in order to obtain a triangulation of a cycle of length
  $k$. Let $n',m',f'$ be the number of vertices, edges and faces of
  $G'$. Every inner vertex of $G'$ has outdegree $3$, there are no
  other inner edges, so the total number of edges of $G'$ is
  $m'=3(n'-k)+k$. All the inner faces have size $3$ and the outer face
  has size $k$, so $2m'=3(f'-1)+k$.  By Euler's formula
  $n'-m'+f'=2$. Combining the three equalities gives $k=3$ and $F$ is
  hence a separating triangle of $G$. This contradicts the choice of
  the root angle, as it should not lie in the strict interior of a
  separating triangle.
\end{proof}

A subgraph of a graph is \emph{spanning} if it is covering all the
vertices. An \emph{Hamiltonian cycle} is a spanning cycle.

\begin{lemma}
\label{lem:ham}
The cycle $C$ is an Hamiltonian cycle of the angle graph, all the
edges of $G$ are marked exactly twice, the subgraph $Q$ of $G^*$ is
spanning, and, if $n\geq 2$, the subgraph $P$ of $G$ is spanning.
\end{lemma}

\begin{proof}
  Suppose for a contradiction that $\overline{C}$ is non empty. By
  Lemma~\ref{lem:noD} and Section~\ref{sec:algo}, the frontier $T$
  between $C$ and $\overline{C}$ is an union of oriented closed walks
  of $P$.  Hence a face of $G$ has either all its angles in $C$ or all
  its angles in $\overline{C}$. Moreover $T$ is a non-empty union of
  oriented closed walk of $P$ that are oriented \cw according to the
  set of faces containing $\overline{C}$ (see the first case of
  Figure~\ref{fig:anglerulesum}). This set does not contain $f_0$
  since $a_0$ is in $f_0$ and $C$.  As in
  Section~\ref{sec:schnyderwoods}, let $\mc{F}$ be the set of
  counterclockwise facial walks of $G$ and $F_0$ be the
  counterclockwise facial walk of $f_0$. Let
  $\mc{F}'=\mc{F}\setminus F_0$, and
  $\mc{F}_{\overline{C}}\subseteq\mc{F}'$ be the set of
  counterclockwise facial walks of the faces containing
  $\overline{C}$. We have
  $\phi(T)=-\sum_{F\in\mc{F}_{\overline{C}}}\phi(F)$.  So $T$ is a
  clockwise non-empty $0$-homologous oriented subgraph
  w.r.t.~$f_0$. This contradicts Lemma~\ref{lem:maximal} and the
  minimality of the orientation w.r.t.~$f_0$. So $\overline{C}$ is
  empty, thus $C$ is Hamiltonian and all the edges of $G$ are marked twice.

  Suppose for a contradiction that $n\geq 2$ and $P$ is not
  spanning. Since the algorithm starts at $v_0$, $P$ is not covering a
  vertex $v$ of $G$ different from $v_0$. Then the angles around $v$
  cannot be visited since by Figure~\ref{fig:anglerulesum} the only
  way to move from an angle of one vertex to an angle of another
  vertex is through an edge of $P$ incident to them. So $P$ is
  spanning.  The proof is similar for $Q$ (note that in this case we
  have $f\geq 2$).
\end{proof}

\begin{lemma}
\label{lem:first}
The first cycle created in $P$ (resp. in $Q$) by the algorithm is
oriented.
\end{lemma}

\begin{proof}
  Let $e$ be the first edge creating a cycle in $P$ while executing
  \aps and consider the steps of \aps before $e$ is added to $P$. So
  $P$ is a tree during all these steps.  For every vertex of $P$ we
  define $P(v)$ the unique path from $v$ to $v_0$ in $P$ (while $P$ is
  empty at the beginning of the execution, we define $P(v_0)=\{v_0\}$).
  By Lemma~\ref{lem:PDconnected}, this path $P(v)$ is an oriented
  path.  We prove the following

  \begin{claim}
\label{cl:leftangles}
Consider a step of the algorithm before $e$ is added to $P$ and where
the algorithm is considering a vertex $v$. Then all the angles around
the vertices of $P$ different from the vertices of $P(v)$ are already
visited.
  \end{claim}

  \begin{proof}
    Suppose by contradiction that there is such a step of the
    algorithm where some angles around the vertices of $P$ different
    from the vertices of $P(v)$ have not been visited. Consider the
    first such step. Then clearly we are not at the beginning of the
    algorithm since $P=P(v)=\{v_0\}$. So at the step just before, the
    conclusion holds and now it does not hold anymore. Clearly at the
    step before we were considering a vertex $v'$ distinct from $v$,
    otherwise $P(v)$ and $P$ have not changed and we have the
    conclusion. So from $v'$ to $v$ we are either in Case~1 or Case~4
    of \apss. If $v$ has been considered by Case~1, then $P(v)$
    contains $P(v')$ and the conclusion holds. If $v$ has been
    considered by Case~4, then since $P$ is a tree, all the angles
    around $v'$ have been considered and $v'$ is the only element of
    $P\setminus P(v)$ that is not in $P\setminus P(v')$. Thus the
    conclusion also holds.
  \end{proof}

  Consider the iteration of \aps where $e$ is added to $P$. The edge
  $e$ is added to $P$ by Case~1, so $e$ is oriented from a vertex $u$
  to a vertex $v$ such that $v$ is already in $P$ or $v$ is the root
  vertex $v_0$. Consider the step of the algorithm just before $u$ is
  added to $P$.  By Claim~\ref{cl:leftangles}, vertex $u$ is not in
  $P\setminus P(v)$ (otherwise $e$ would have been considered before
  and it would be a stem). So $u\in P(v)$ and $P(v)\cup\{e\}$ induces
  an oriented cycle of $G$. The proof is similar for $Q$.
\end{proof}

\begin{lemma}
\label{lem:unicellular}
  $P$ is a spanning unicellular map of $G$ and $Q$ is a spanning tree
  of $G^*$. Moreover one is the dual of the complement of the other.
\end{lemma}
\begin{proof}
  Suppose that $Q$ contains a cycle, then by Lemma~\ref{lem:first} it
  contains an oriented cycle of $G^*$. This cycle is contractible by
  Lemma~\ref{lem:incomingedges}. Recall that by Lemma~\ref{lem:ham},
  $C$ is an Hamiltonian cycle, moreover it does not cross $Q$, a
  contradiction.  So $Q$ contains no cycle and is a tree.

  By Lemma~\ref{lem:ham}, all the edges of $G$ are marked at the
  end. So every edge of $G$ is either in $P$ or its dual in $Q$ (and
  not both). Thus $P$ and $Q$ are the dual of the complement of each
  other. So $P$ is the dual of the complement of a spanning tree of
  $G^*$. Thus $P$ is a spanning unicellular map of $G$.
\end{proof}

Theorem~\ref{th:uni} is then a direct reformulation of
Lemma~\ref{lem:unicellular} by the definition of $P$ and $Q$:

A toroidal unicellular map on $n$ vertices has exactly $n+1$ edges:
$n-1$ edges of a tree plus $2$ edges corresponding to the size of a
basis of the homology (i.e. plus $2g$ in general for an oriented
surface of genus $g$).  Thus a consequence of Theorem~\ref{th:uni} is
that the obtained unicellular map $U$ has exactly $n$ vertices, $n+1$
edges and $2n-1$ stems since the total number of edges is $3n$.  The
orientation of $G$ induces an orientation of $U$ such that the stems
are all outgoing, and such that while walking \cw around the unique
face of $U$ from $a_0$, the first time an edge is met, it is oriented
\ccw according to this face, see Figure~\ref{fig:hexasquare} where all the
tree-like parts and stems are not represented.  There are two types
of toroidal unicellular maps depicted on
Figure~\ref{fig:hexasquare}. Two cycles of $U$ may intersects either
on a single vertex (square case) or on a path (hexagonal case).  The
square can be seen as a particular case of the hexagon where one side
has length zero and thus the two corners of the hexagon are
identified.

\begin{figure}[!h]
\center
\begin{tabular}{cc}
\includegraphics[scale=0.5]{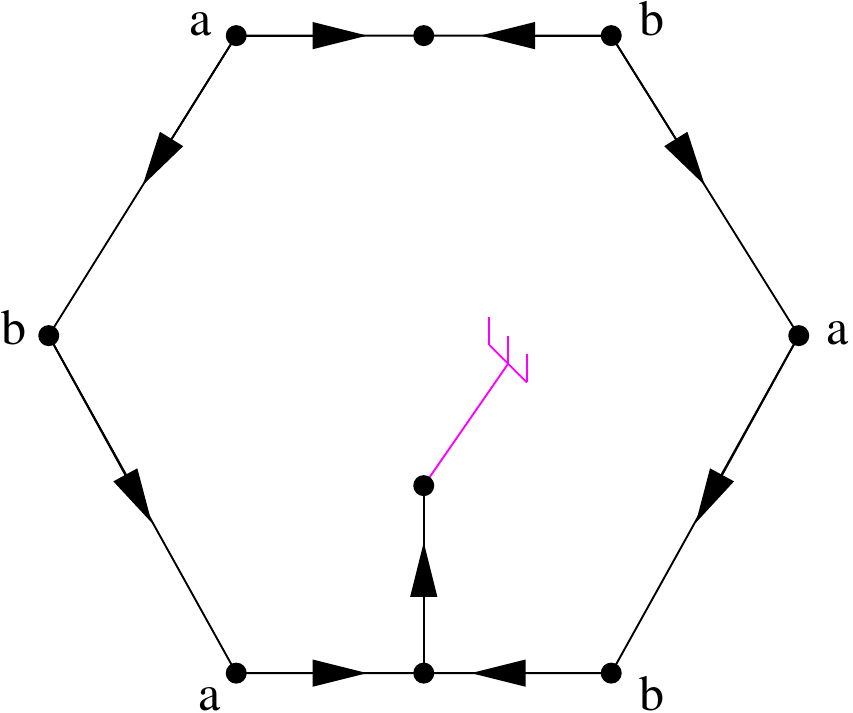} \ \ & \ \
\includegraphics[scale=0.5]{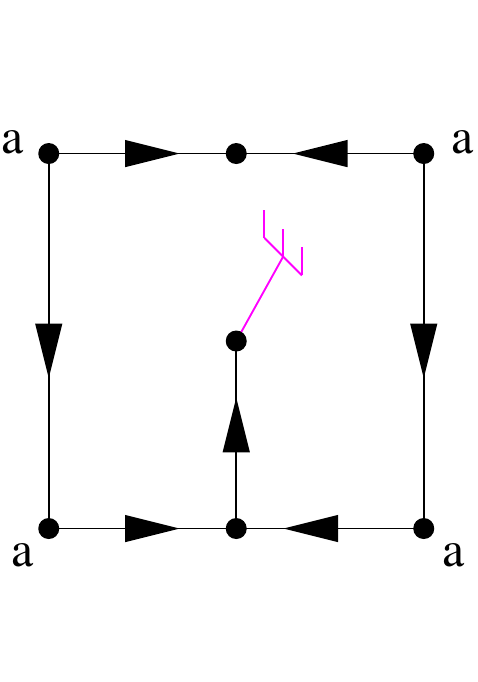}\\
Hexagon \ \ & \ \ Square
\end{tabular}

\caption{The two types of rooted toroidal unicellular maps.}
\label{fig:hexasquare}
\end{figure}

On Figure~\ref{fig:contre-exemple}, we give several examples of
executions of \aps on minimal $3$-orientations.  These examples show
how important is the choice of the minimal HTC Schnyder wood in order
to obtain Theorem~\ref{th:uni}. In particular, the third example shows
that \aps can visit all the angles of the triangulation (i.e. the
cycle in the angle graph is Hamiltonian) without outputting an
unicellular map.

\begin{figure}[h!]
\center
\begin{tabular}[c]{cc}
\includegraphics[scale=0.5]{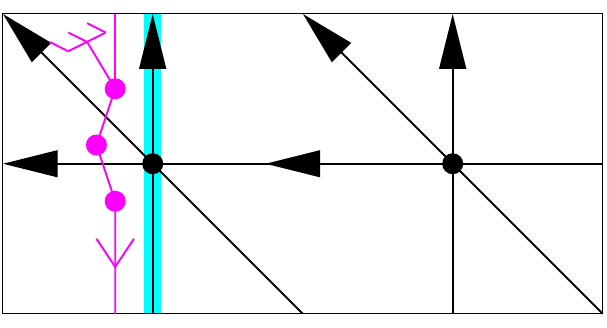} \ \ \ \ & \ \ \ \
\includegraphics[scale=0.5]{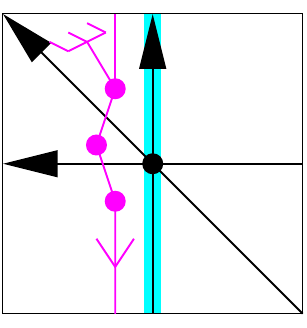}\\
(1) \ \ \ \ & \ \ \ \ (2)
\end{tabular}

\ \\

\ \\

 \begin{tabular}[c]{c}
 \includegraphics[scale=0.5]{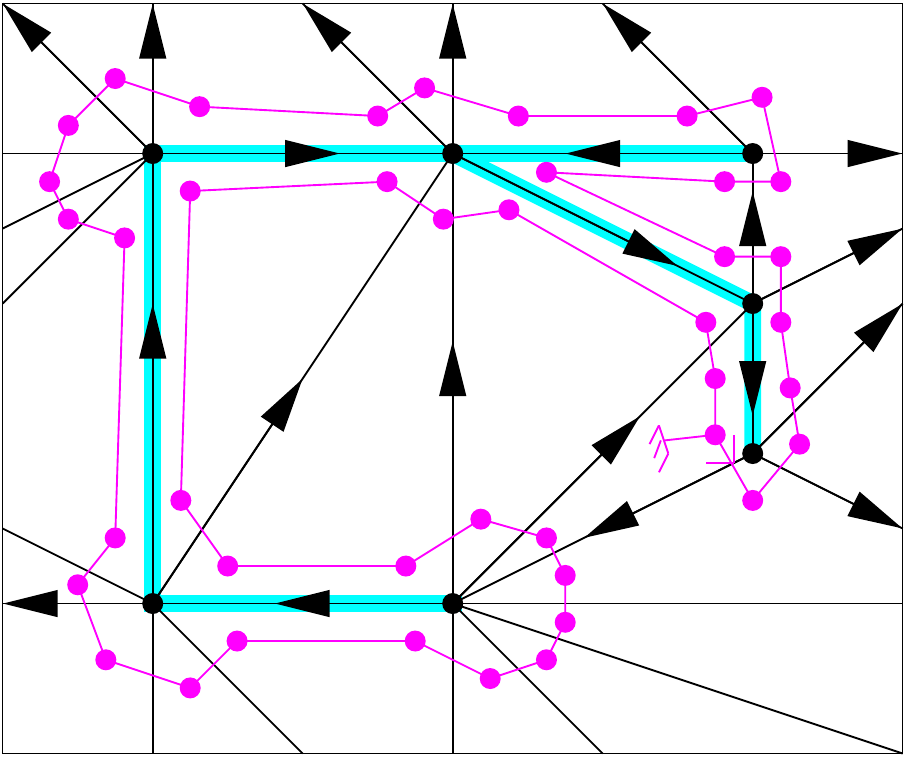} \\
 (3)
 \end{tabular}

\caption{Examples of minimal $3$-orientations that are not HTC
  Schnyder woods and where \aps respectively: $(1)$ does not visit all
  the vertices, $(2)$ visits all the vertices but not all the angles, and
  $(3)$ visits all the angles but does not output an unicellular map.}
\label{fig:contre-exemple}
\end{figure}

Note that the orientations of Figure~\ref{fig:contre-exemple} are not
Schnyder woods. One may wonder if the fact of being a Schnyder wood is
of any help for our method. This is not the case since there are
examples of minimal Schnyder woods that are not HTC and where \aps
does not visit all the vertices. One can obtain such an example by
replicating $3$ times horizontally and then $3$ times vertically the
second example of Figure~\ref{fig:contre-exemple} to form a
$3\times 3$ tiling and starts \aps from the same root angle.
Conversely, there are minimal Schnyder woods that are not HTC where
\aps does output a toroidal unicellular map covering all the edges (the
Schnyder wood of Figure~\ref{fig:gamma0glue} can serve as an example
while starting from an angle of the only face oriented \cww).

\section{Recovering the original triangulation}
\label{sec:close}

This section is dedicated to show how to recover the original
triangulation from the output of \apss. The method is very similar
to~\cite{PS06} since like in the plane the output has only one face
that is homeomorphic to an open disk (i.e. a tree in the plane and an
unicellular map in general). 

\begin{theorem}
\label{th:recover}
Consider a toroidal triangulation $G$, a root angle $a_0$ that is not
in the strict interior of a separating triangle and the orientation of
the edges of $G$ corresponding to the minimal HTC Schnyder wood
w.r.t.~the root face $f_0$ containing $a_0$. From the output $U$ of
\aps applied on $(G,a_0)$ one can reattach all the stems to obtain $G$
by starting from the root angle $a_0$ and walking along the face of
$U$ in \ccw order (according to this face): each time a stem is met,
it is reattached in order to create a triangular face on its left
side.
\end{theorem}

Theorem~\ref{th:recover} is illustrated on
Figure~\ref{fig:reconstruction} where one can check that the obtained
toroidal triangulation is $K_7$ (like on the input of
Figure~\ref{fig:tore-example}).

\begin{figure}[h!]
\center
\begin{tabular}{cc}
\includegraphics[scale=0.5]{tore-tri-exe-4---.pdf}
 & 
\includegraphics[scale=0.5]{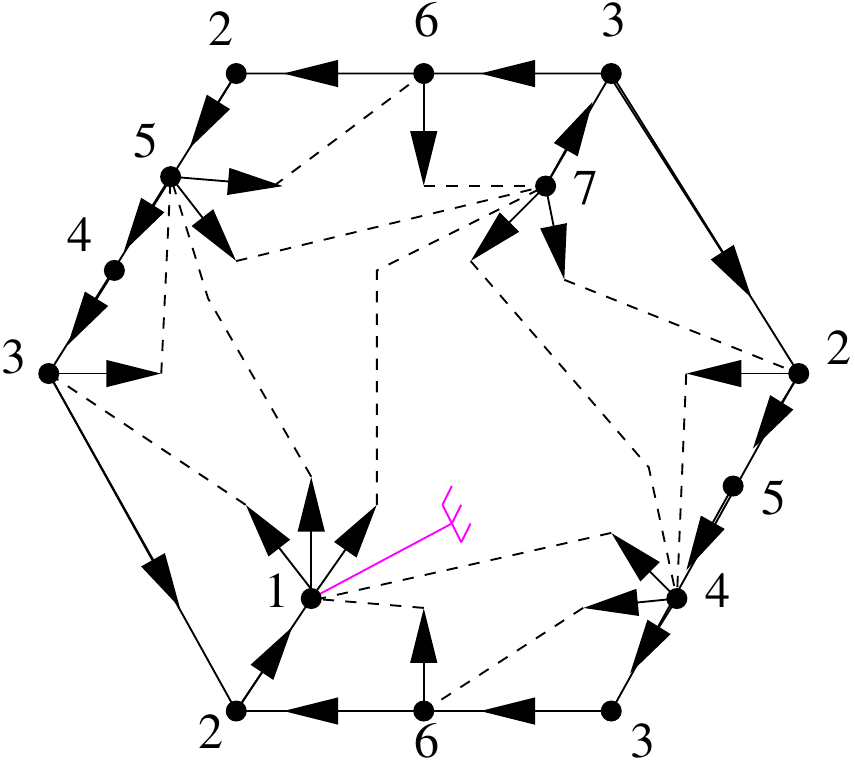}\\
Output & Recovering \\
\end{tabular}
\caption{Example of how to recover the original toroidal
  triangulation $K_7$ from the output of \aps.}
\label{fig:reconstruction}
\end{figure}

In fact in this section we define a method, more general than the one
described in Theorem~\ref{th:recover}, that will be useful in next
sections.

Let $\mathcal U_r(n)$ denote the set of toroidal unicellular maps $U$
rooted on a particular angle $a_0$, with exactly $n$ vertices, $n+1$
edges and $2n-1$ stems satisfying the following: every vertex has
exactly $2$ stems, except the root vertex $v_0$ that has $1$ more
stem, and if the map is hexagonal, the two corners that have $1$ less
stem each, and if the map is square, the only corner that has 
$2$ less stems (if the root vertex is a corner we simply combine the more
and less).  Note that the output $U$ of \aps given by
Theorem~\ref{th:uni} is an element of $\mathcal U_r(n)$.

Similarly to the planar case~\cite{PS06}, we define a general way to
reattached step by step all the stems of an element $U$ of
$\mathcal U_r(n)$.  Let $U_0=U$, and, for $1\leq k \leq 2n-1$, let
$U_{k}$ be the map obtained from $U_{k-1}$ by reattaching one of its
stem (we explicit below which stem is reattached and how). The
\emph{special face of $U_0$} is its only face. For
$1\leq k \leq 2n-1$, the \emph{special face of $U_{k}$} is the face on
the right of the stem of $U_{k-1}$ that is reattached to obtain
$U_{k}$.  For $0\leq k\leq 2n-1$, the border of the special face of
$U_k$ consists of a sequence of edges and stems. We define an
\emph{admissible triple} as a sequence $(e_1,e_2,s)$, appearing in
\ccw order along the border of the special face of $U_k$, such that
$e_1=(u,v)$ and $e_2=(v,w)$ are edges of $U_k$ and $s$ is a stem
attached to $w$. The \emph{closure} of the admissible triple consists
in attaching $s$ to $u$, so that it creates an edge $(w,u)$ oriented
from $w$ to $u$ and so that it creates a triangular face $(u,v,w)$ on
its left side.  The \emph{complete closure} of $U$ consists in closing
a sequence of admissible triple, i.e.  for $1\leq k \leq 2n-1$, the
map $U_{k}$ is obtained from $U_{k-1}$ by closing any admissible
triple.

Note that, for $0\leq k\leq 2n-1$, the special face of $U_k$ contains
all the stems of $U_k$. The closure of a stem reduces the number of
edges on the border of the special face and the number of stems by
$1$. At the beginning, the unicellular map $U_0$ has $n+1$ edges and
$2n-1$ stems. So along the border of its special face, there are
$2n+2$ edges and $2n-1$ stems. Thus there is exactly three more edges
than stems on the border of the special face of $U_0$ and this is
preserved while closing stems. So at each step there is necessarily at
least one admissible triple and the sequence $U_k$ is well defined.
Since the difference of three is preserved, the special face of
$U_{2n-2}$ is a quadrangle with exactly one stem. So the reattachment
of the last stem creates two faces that are triangles and at the end
$U_{2n-1}$ is a toroidal triangulation.  Note that at a given step
there might be several admissible triples but their closure are
independent and the order in which they are performed does not modify
the obtained triangulation $U_{2n-1}$.

We now apply the closure method to our particular case.  Consider a
toroidal triangulation $G$, a root angle $a_0$ that is not in the
strict interior of a separating triangle and the orientation of the
edges of $G$ corresponding to the minimal HTC Schnyder wood w.r.t.~the
root face $f_0$. Let $U$ be the output of \aps
applied on $(G,a_0)$.

\begin{lemma}
\label{lem:stemright}
When a stem of $U$ is reattached to form the corresponding edge of
$G$, it splits the (only) face of $U$ into two faces. The root angle
of $U$ is in the face that is on the right side of the stem.
\end{lemma}

\begin{proof}
  By Lemma~\ref{lem:ham}, the execution of \aps corresponds to an
  Hamiltonian cycle $C=(a_0, \ldots, a_{2m},a_0)$ in the angle graph
  of $G$. Thus $C$ defines a total order $<$ on the angles of $G$
  where $a_i< a_j$ if and only if $i< j$.  Let us consider now the
  angles on the face of $U$.  Note that such an angle corresponds to
  several angles of $G$, that are consecutive in $C$ and that are
  separated by a set of incoming edges of $G$ (those incoming edges
  corresponding to stems of $U$).  Thus the order on the angles of $G$
  defines automatically an order on the angles of $U$.  The angles of
  $U$ considered in \cw order along the border of its face, starting
  from the root angle, correspond to a sequence of strictly
  increasing angles for $<$.

  Consider a stem $s$ of $U$ that is reattached to form an edge $e$ of
  $G$.  Let $a_s$ be the angle of $U$ that is situated just before $s$
  (in \cw order along the border of the face of $U$) and $a'_s$ be the
  angle of $U$ where $s$ should be reattached.  If $a'_s< a_s$, then
  when \aps consider the angle $a_s$, the edge corresponding to $s$ is
  already marked and we are not in Case~2 of \aps. So $a_s< a'_s$ and
  $a_0$ is on the right side of $s$.
\end{proof}

Recall that $U$ is an element of $\mathcal U_r(n)$ so we can apply on
$U$ the complete closure procedure described above.  We use the same
notation as before, i.e. let $U_0=U$ and for $1\leq k \leq 2n-1$, the
map $U_{k}$ is obtained from $U_{k-1}$ by closing any admissible
triple.  The following lemma shows that the triangulation obtained by
this method is $G$:

\begin{lemma}
\label{lem:stemstep}
  The complete closure of $U$ is $G$,  i.e. $U_{2n-1}=G$.
\end{lemma}

\begin{proof}
  We prove by induction on $k$ that every face of $U_k$ is a face of
  $G$, except for the special face.  This is true for $k=0$ since
  $U_0=U$ has only one face, the special face.  Let
  $0\leq k\leq 2n-2$, and suppose by induction that every non-special
  face of $U_k$ is a face of $G$.  Let $(e_1,e_2,s)$ be the admissible
  triple of $U_k$ such that its closure leads to $U_{k+1}$, with
  $e_1=(u,v)$ and $e_2=(v,w)$. The closure of this triple leads to a
  triangular face $(u,v,w)$ of $U_{k+1}$. This face is the only
  ``new'' non-special face while going from $U_k$ to $U_{k+1}$.

  Suppose, by contradiction, that this face $(u,v,w)$ is not a face of
  $G$.  Let $a_v$ (resp. $a_w$) be the angle of $U_k$ at the special
  face, between $e_1$ and $e_2$ (resp. $e_2$ and $s$).  Since $G$ is a
  triangulation, and $(u,v,w)$ is not a face of $G$, there exists at
  least one stem of $U_k$ that should be attached to $a_v$ or $a_w$ to
  form a proper edge of $G$. Let $s'$ be such a stem that is the
  nearest from $s$. In $G$ the edges corresponding so $s$ and $s'$
  should be incident to the same triangular face. Let $x$ be the
  origin of the stem $s'$.  Let $z\in \{v,w\}$ such that $s'$ should
  be reattached to $z$.  If $z=v$, then $s$ should be reattached to
  $x$ to form a triangular face of $G$. If $z=w$, then $s$ should be
  reattached to a common neighbor of $w$ and $x$ located on the border
  of the special face of $U_k$ in \ccw order between $w$ and $x$. So
  in both cases $s$ should be reattached to a vertex $y$ located on
  the border of the special face of $U_k$ in \ccw order between $w$
  and $x$ (with possibly $y=x$). To summarize $s$ goes from $w$ to $y$
  and $s'$ from $x$ to $z$, and $z,x,y,w$ appear in \cw order along
  the special face of $U_k$. By Lemma~\ref{lem:stemright}, the root
  angle is on the right side of both $s$ and $s'$, this is not
  possible since their right sides are disjoint, a contradiction.

  So for $0\leq k\leq 2n-2$, all the non-special faces of $U_k$ are
  faces of $G$. In particular every face of $U_{2n-1}$ except one is a
  face of $G$. Then clearly the (triangular) special face of
  $U_{2n-1}$ is also a face of $G$, hence $U_{2n-1}=G$.
\end{proof}

Lemma~\ref{lem:stemstep} shows that one can recover the original
triangulation from $U$ with any sequence of admissible triples
that are closed successively. This does not explain how to find the
admissible triples efficiently. In fact the root angle can be used to
find a particular admissible triple of $U_k$:

\begin{lemma}
\label{lem:firststem}
For $0\leq k\leq 2n-2$, let $s$ be the first stem met while walking
\ccw from $a_0$ in the special face of $U_k$. Then before $s$, at
least two edges are met and  the last two of these edges form an
admissible triple with $s$.
\end{lemma}

\begin{proof}
  Since $s$ is the first stem met, there are only edges that are met
  before $s$. Suppose by contradiction that there is only zero or one
  edge met before $s$. Then the reattachment of $s$ to form the
  corresponding edge of $G$ is necessarily such that the root angle is
  on the left side of $s$, a contradiction to
  Lemma~\ref{lem:stemright}. So at least two edges are met before $s$
  and the last two of these edges form an admissible triple with $s$.
\end{proof}

Lemma~\ref{lem:firststem} shows that one can reattach all the stems
by walking once along the face of $U$ in \ccw order. Thus we obtain
Theorem~\ref{th:recover}.

Note that $U$ is such that the complete closure procedure described
here never \emph{wraps over the root angle}, i.e. when a stem is
reattached, the root angle is always on its right side (see
Lemma~\ref{lem:stemright}). The property of never wrapping over the
root angle is called \emph{balanced} in~\cite{AP13}.  Let
$\mathcal U_{r,b}(n)$ denote the set of elements of $\mathcal U_r(n)$
that are balanced. So the output $U$ of \aps given by
Theorem~\ref{th:uni} is an element of $\mathcal U_{r,b}(n)$.  We
exhibit in Section~\ref{sec:bijection} a bijection between
appropriately rooted toroidal triangulations and a particular subset
of $\mathcal U_{r,b}(n)$.

The possibility to close admissible triples in any order to recover
the original triangulation is interesting comparing to the simpler
method of Theorem~\ref{th:recover} since it enables to recover the
triangulation even if the root angle is not given.  This property is
used in Section~\ref{sec:bij2} to obtain a bijection between toroidal
triangulations and some unrooted unicellular maps.

Moreover if the root angle is not given, then one can simply start
from any angle of $U$, walk twice around the face of $U$ in \ccw order
and reattached all the admissible triples that are encountered along
this walk. Walking twice ensure that at least one complete round is
done from the root angle. Since only admissible triples are
considered, we are sure that no unwanted reattachment is done during
the process and that the final map is $G$. This enables to reconstruct
$G$ in linear time even if the root angle is not known. This property
is used in Section~\ref{sec:coding}.

\section{Optimal encoding}
\label{sec:coding}
The results presented in the previous sections allow us to generalize
the encoding of planar triangulations, defined by Poulalhon and
Schaeffer~\cite{PS06}, to triangulations of the torus. The
construction is direct and it is hence really different from the one
of~\cite{CFL10} where triangulations of surfaces are cut in order to
deal with planar triangulations with boundaries. Here we encode the
unicellular map outputted by \aps by a plane rooted tree with
$n$ vertices and with exactly two stems attached to each vertex, plus
$O(\log(n))$ bits.  As in~\cite{CFL10}, this encoding is
asymptotically optimal and uses approximately $3.25 n$ bits. The
advantage of our method is that it can be implemented in linear
time. Moreover we believe that our encoding gives a better
understanding of the structure of triangulations of the torus. It is
illustrated with news bijections that are obtained in
Sections~\ref{sec:bijection} and~\ref{sec:bij2}.

Consider a toroidal triangulation $G$, a root angle $a_0$ that is not
in the strict interior of a separating triangle and the orientation of
the edges of $G$ corresponding to the minimal HTC Schnyder wood
w.r.t.~the root face $f_0$. Let $U$ be the output of \aps applied on
$(G,a_0)$. As already mentioned at the end of Section~\ref{sec:close},
to retrieve the triangulation $G$ one just needs to know $U$ without
the information of its root angle (by walking twice around the face of
$U$ in \ccw order and reattached all the admissible triples that are
encountered along this walk, one can recover $G$). Hence to encode
$G$, one just has to encode $U$ without the position of the root angle
around the root vertex (see Figure~\ref{fig:Coding}.(a)).

By Lemma~\ref{lem:PDconnected}, the unicellular map $U$ contains a
spanning tree $T$ which is oriented from the leaves to the root
vertex. The tree $T$ contains exactly $n-1$ edges, so there is exactly
$2$ edges of $U$ that are not in $T$. We call these edges the
\emph{special edges} of $U$. We cut these two special edges to
transform them into stems of $T$ (see Figures~\ref{fig:Coding}.(a)
and~(b)). We keep the information of where are the special stems in
$T$ and on which angle of $T$ they should be reattached. This
information can be stored with $O(\log(n))$ bits.  One can recover $U$
from $T$ by reattaching the special stems in order to form
non-contractible cycles with $T$ (see Figure~\ref{fig:Coding}.(c)).

\begin{figure}[h!]
\center
\begin{tabular}{ccc}
\includegraphics[scale=0.3]{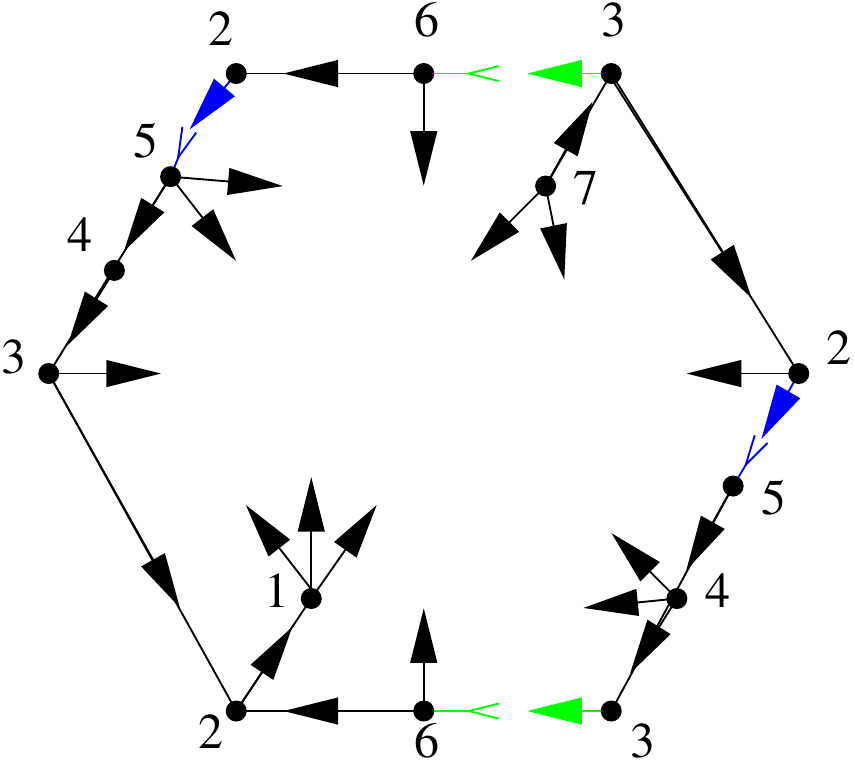} \ & \
\includegraphics[scale=0.3]{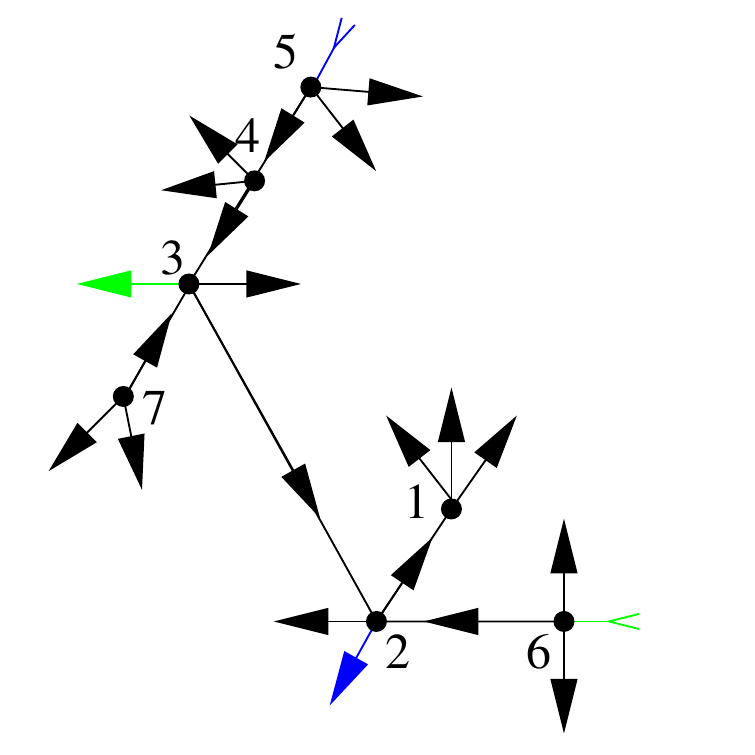} \ & \ 
\includegraphics[scale=0.3]{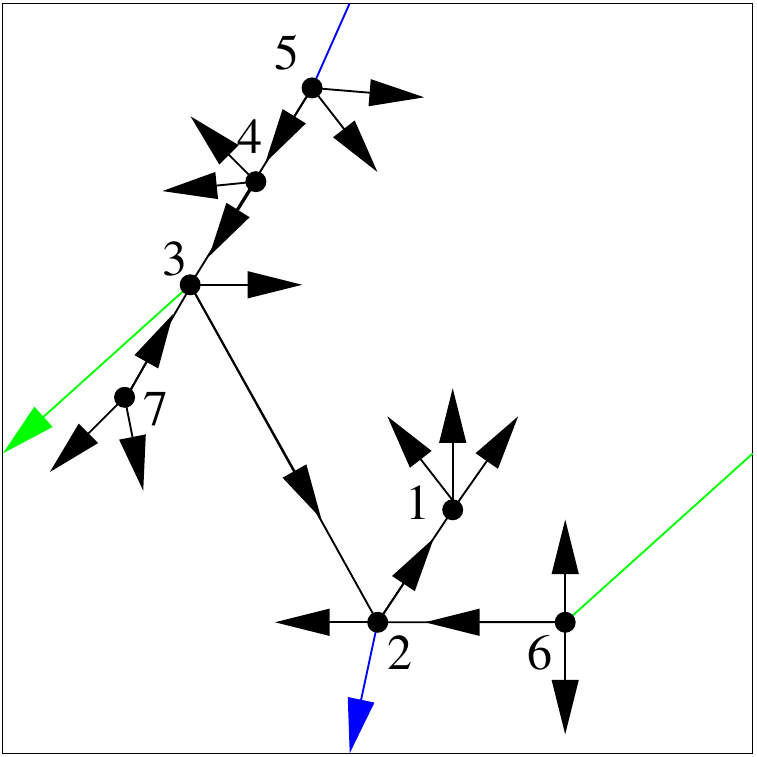} \\
(a) \ &\ (b) \ &\ (c) \\
\end{tabular}
\caption{From unicellular maps to trees with special stems and back.}
\label{fig:Coding}
\end{figure}

So $T$ is a plane tree on $n$ vertices, each vertex having $2$ stems
except the root vertex $v_0$ having three stems. Choose any stem $s_0$
of the root vertex, remove it and consider that $T$ is rooted at the
angle where $s_0$ should be attached. The information of the root
enable to put back $s_0$ at its place. So now we are left with a
rooted plane tree $T$ on $n$ vertices where each vertex has exactly
$2$ stems (see Figure~\ref{fig:CodingTree}.(a)).

This tree $T$ can easily be encoded by a binary word on $6n-2$ bits:
that is, walking in \ccw order around $T$ from the root angle, writing
a ``1'' when going down along $T$, and a ``0'' when going up along $T$
(see Figure~\ref{fig:CodingTree}.(a)). As in~\cite{PS06}, one can
encode $T$ more compactly by using the fact that each vertex has
exactly two stems.  Thus $T$ is encoded by a binary word on $4n-2$
bits: that is, walking in \ccw order around $T$ from the root angle,
writing a ``1'' when going down along an edge of $T$, and a ``0'' when
going up along an edge or along a stem of $T$ (see
Figure~\ref{fig:CodingTree}.(b) where the ``red 1's'' of
Figure~\ref{fig:CodingTree}.(a) have been removed).  Indeed there is
no need to encode when going down along stems, this information can be
retrieved afterward. While reading the binary word to recover $T$,
when a ``0'' is met, we should go up in the tree, except if the
vertex that we are considering does not have already its two stems,
then in that case we should create a stem (i.e. add a ``red 1'' before
the ``0''). So we are left with a binary word on $4n-2$ bits with
exactly $n-1$ bits ``1'' and $3n-1$ bits ``0''.

\begin{figure}[h!]
\center
\begin{tabular}{cc}
\includegraphics[scale=0.4]{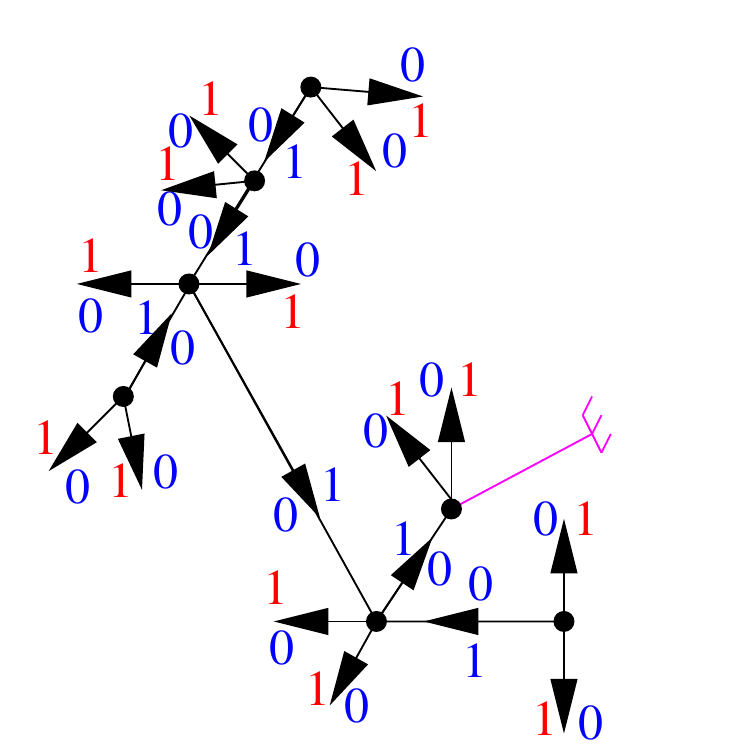}\ & \
\includegraphics[scale=0.4]{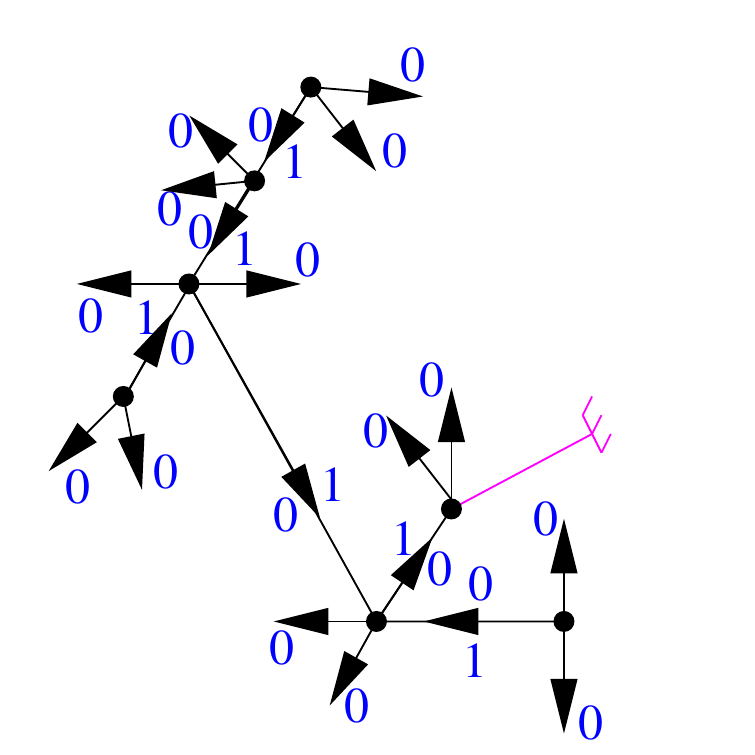}\\

{\color{blue}{\color{red}1}0{\color{red}1}011{\color{red}1}011{\color{red}1}0{\color{red}1}00{\color{red}1}0{\color{red}1}00{\color{red}1}01{\color{red}1}0{\color{red}1}000{\color{red}1}0{\color{red}1}01{\color{red}1}0{\color{red}1}000}  \ & \
{\color{blue}00110110000000100000010000} \\
(a) \ &\ (b) \\
\end{tabular}
\caption{Encoding a rooted tree with two stems at each vertex.}
\label{fig:CodingTree}
\end{figure}

Similarly to~\cite{PS06}, using \cite[Lemma~7]{BGH03}, this word can
then be encoded with a binary word of length
$\log_2\binom{4n-2}{n-1}+o(n)\sim n\, \log_2(\frac{256}{27})\approx
3,2451\,n$
bits. Thus we have the following theorem whose linearity is discussed
in Section~\ref{sec:complexity}:

\begin{theorem}
\label{th:encoding}
  Any toroidal triangulation on $n$ vertices, can be encoded with a
  binary word of length $3.2451 n +o(n)$ bits, the encoding and
  decoding being linear in $n$.
\end{theorem}

\section{Linear complexity}
\label{sec:complexity}
In this section we show that the encoding method described in this
paper, that is encoding a toroidal triangulation via an unicellular
map and recovering the original triangulation, can be performed in
linear time.  The only difficulty lies in providing \aps with the
appropriate input it needs in order to apply
Theorem~\ref{th:uni}. Then clearly the execution of \apss, the
encoding phase and the recovering of the triangulation are linear.
Thus we have to show how one can find in linear time a root angle
$a_0$ that is not in the strict interior of a separating triangle, as
well as the minimal HTC Schnyder wood w.r.t.~the root face $f_0$.

Consider a toroidal triangulation $G$. Let us see how one can build a
Schnyder wood of $G$ in linear time.  The \emph{contraction} of a
non-loop-edge $e$ of $G$ is the operation consisting of continuously
contracting $e$ until merging its two ends, as shown on
Figure~\ref{fig:contraction}. Note that only one edge of each pair of
edges forming a contractible 2-cycle is preserved (edges $e_{wx}$ and
$e_{wy}$ on the figure).

\begin{figure}[h]
\center
\includegraphics[scale=0.4]{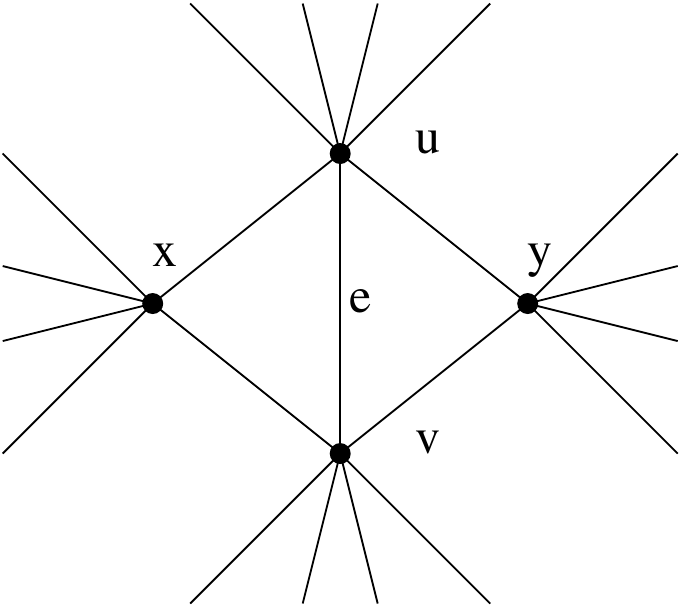}
 \hspace{3em} 
\includegraphics[scale=0.4]{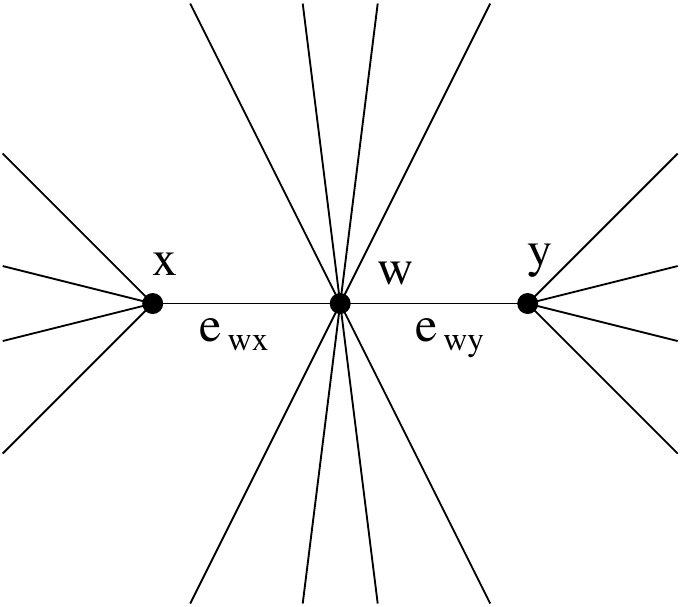}
\caption{The contraction operation.}
\label{fig:contraction}
\end{figure}

An edge $uv$ is said to be \emph{contractible} if it is not a loop and
if it is not on a separating triangle (i.e. if after contracting $uv$
one obtains a triangulation that is still without contractible 1- or
2-cycles).  In~\cite{GL13} the existence of crossing Schnyder wood is
proved by contraction. Unfortunately this proof cannot easily be
transformed into a linear algorithm because of the crossing property
that has to be maintained during the contraction process. Nevertheless
we use contractions to obtain non-necessarily crossing Schnyder woods.
If the triangulation obtained after contracting a contractible edge
admits a Schnyder wood it is then easy to obtain a Schnyder wood of
$G$. The rules for decontracting an edge in the case of toroidal
triangulations are depicted on \cite[Figure~21]{GL13} where for each
case one can choose any of the proposed colorings.  For any toroidal
triangulation, one can find contractible edges until the toroidal map
has only one vertex (see~\cite{Moh96}). A Schnyder wood of the
toroidal map on one vertex is depicted on the right of
Figure~\ref{fig:orientation}. Thus one can obtain a Schnyder wood of
any toroidal triangulation by this process. Nevertheless, to maintain
linearity we have to be more precise since it is not trivial to find
contractible edges.

Consider an edge $uv$ of $G$ with incident faces $uvx$ and $vuy$ such
that these vertices appear in \cw order around the corresponding face
(so we are in the situation of Figure~\ref{fig:contraction}). If $u$
and $v$ have more common neighbors, then consider their second common
neighbor going clockwise around $u$ from $uv$ (the first one being
$x$, and the last being $y$) and call it $x'$. Call $y'$ their second
common neighbor going counterclockwise around $u$ from $uv$. Then
either $uvx'$ or $uvy'$ is a separating triangle or edge $uv$ is
contractible. We consider these two cases below:
\begin{itemize}
\item If $uv$ is contractible, then it is contracted and we apply the
  procedure recursively to obtain a Schnyder wood of the contracted
  graph. Then we update the Schnyder wood as described above. Note
  that this update is done in constant time.
\item If $uvx'$ (resp. $uvy'$) is a separating triangle, one can
  remove its interior, recursively obtain a toroidal Schnyder wood of
  the remaining toroidal triangulation, build a planar Schnyder wood
  of the planar triangulation inside $uvx'$ (resp. $uvy'$), and then
  superimpose the two (by eventually permuting the colors) to obtain a
  Schnyder wood of the whole graph. Note that computing a planar
  Schnyder wood can be done in linear time using a canonical ordering
  (see~\cite{Kan96}).
\end{itemize}

The difficulty here is to test if $uvx'$ or $uvy'$ are contractible
triangles.  For that purpose, one first need to compute a basis
$(B_1,B_2)$ for the homology.  Consider a spanning tree of the dual
map $G^*$. The map obtained from $G$ by removing those edges is
unicellular, and removing its treelike parts one obtains two cycles
$(B_1,B_2)$ (intersecting on a path with at least one vertex) that
form a basis for the homology. This can be computed in linear time for
$G$ and then updated in constant time when some edge is contracted or
when the interior of some separating triangle is removed.  Then a
closed walk $W$, given with an arbitrary orientation, is contractible
if and only if $W$ crosses $B_i$ from right to left as many times as
$W$ crosses $B_i$ from left to right, for $i\in\{1,2\}$.  This test is
linear in $|W|$ hence constant time for the triangles $uvx'$ and
$uvy'$. Vertex $u$ is fixed during the whole process so the total
running time to compute a Schnyder wood of $G$ is linear.

From this Schnyder wood, one can compute in linear time a root angle
$a_0$ not in the strict interior of a separating triangle. First note
that in a $3$-orientation of a toroidal triangulation, the edges that
are inside a separating triangle and that are incident to the three
vertices on the border are all oriented toward these three vertices
by Euler's formula. Thus an oriented non-contractible cycle cannot
enter in the interior of a separating triangle. Now follow any
oriented monochromatic path of the Schnyder wood and stop the first
time this path is back to a previously met vertex $v_0$. The end
of this path forms an oriented monochromatic cycle $C$ containing
$v_0$.  If $C$ is contractible then Euler's formula is violated in the
contractible region.  Thus $C$ is an oriented non-contractible cycle
and cannot contain some vertices that are in the interior of a
separating triangle. So $v_0$ is not in the interior of a separating
triangle and we can choose as root angle $a_0$ any angle incident to
$v_0$.

In~\cite[Section 9]{GKL15} it is proved how one can transform any
$3$-orientation (hence a Schnyder wood) of a toroidal triangulation
into an half-crossing (hence HTC) Schnyder wood.  The method consists
in computing a so called ``middle-path'' (a directed path where the
next edge chosen is the one leaving in the ``middle'') and reversing
some non-contractible ``middle-cycles''. Clearly the method is linear
even if not explicitly mentioned in~\cite{GKL15}.  Let $D_0$ be the
corresponding obtained orientation of $G$.

It remains to compute the minimal HTC Schnyder wood w.r.t.~the root
face $f_0$. There is a generic known
method~\cite{Meu} (see also~\cite[p.23]{UecPHD}) to compute in linear
time a minimal $\alpha$-orientation of a planar map as soon as an
$\alpha$-orientation is given. This method also works on oriented
surfaces and can be applied to obtain the minimal HTC Schnyder wood in
linear time. We explain the method briefly below.

It is much simpler to compute the minimal orientation $D_{\min}$
homologous to $D_0$ in a dual setting.  The first observation to make
is that two orientations $D_1, D_2$ of $G$ are homologous if and only
if there dual orientations $D^*_1, D^*_2$ of $G^*$ are equivalent up
to reversing some directed cuts. Furthermore $D_1\le_{f_0} D_2$ if and
only if $D^*_1$ can be obtained from $D^*_2$ by reversing directed
cuts oriented from the part containing $f_0$.  Let us compute
$D^*_{\min}$ which is the only orientation of $G^*$, obtained from
$D^*_0$ by reversing directed cuts, and without any directed cut
oriented from the part containing $f_0$.  For this, consider the
orientation $D^*_0$ of $G^*=(F,E^*)$ and compute the set
$X\subseteq F$ of vertices of $G^*$ that have an oriented path toward
$f_0$. Then $(X,F\setminus X)$ is a directed cut oriented from the
part containing $f_0$ that one can reverse. Then update the set of
vertices that can reach $f_0$ and go on until $X=F$.  It is not
difficult to see that this can be done in linear time. Thus we obtain
the minimal HTC Schnyder wood w.r.t.~$f_0$ in linear time.

\section{Bijection with rooted unicellular maps}
\label{sec:bijection}

Given a toroidal triangulation $G$ with a root angle $a_0$, we have
defined a unique associated orientation: the minimal HTC Schnyder wood
w.r.t.~the root face $f_0$. Suppose that $G$ is
oriented according to the minimal HTC Schnyder wood.  If $a_0$ is not
in the strict interior of a separating triangle then
Theorems~\ref{th:uni} and~\ref{th:recover} show that the execution of
\aps on $(G,a_0)$ gives a toroidal unicellular map with stems from
which one can recover the original triangulation. Thus there is a
bijection between toroidal triangulations rooted from an appropriate
angle and their image by \apss. The goal of this section is to
describe this image.

Recall from Section~\ref{sec:close} that the output of \aps on
$(G,a_0)$ is an element of $\mathcal U_{r,b}(n)$.  One may hope that
there is a bijection between toroidal triangulations rooted from an
appropriate angle and $\mathcal U_{r,b}(n)$ since this is how it works
in the planar case. Indeed, given a planar triangulation $G$, there is
a unique orientation of $G$ (the minimal Schnyder wood) on which
\apss, performed from an outer angle, outputs a tree covering all the
edges of the graph. In the toroidal case, things are more complicated
since the behavior of \aps on minimal HTC Schnyder woods does not
characterize such orientations.

Figure~\ref{fig:twounicellular} gives an example of two
(non-homologous) orientations of the same triangulation that are both
minimal w.r.t.~the same root face. For these two orientations, the
execution of \aps from the same root angle gives two different
elements of $\mathcal U_{r,b}(2)$ (from which the original
triangulation can be recovered by the method of
Theorem~\ref{th:recover}).  Thus we have to exhibit a particular
property of HTC Schnyder woods that can be used to characterize which
particular subset of $\mathcal U_{r,b}(n)$ is in bijection with
appropriately rooted toroidal triangulations.

\begin{figure}[h!]
\center
\begin{tabular}[c]{cc}
\includegraphics[scale=0.5]{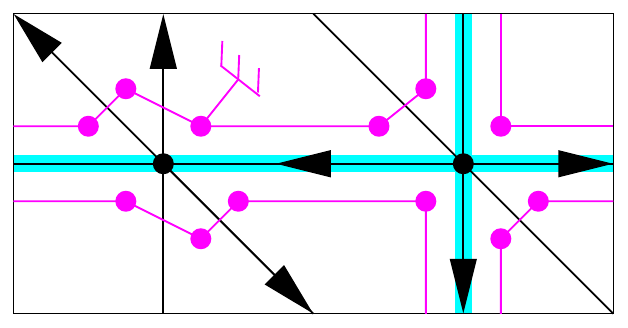}  \ \ & \ \
\includegraphics[scale=0.5]{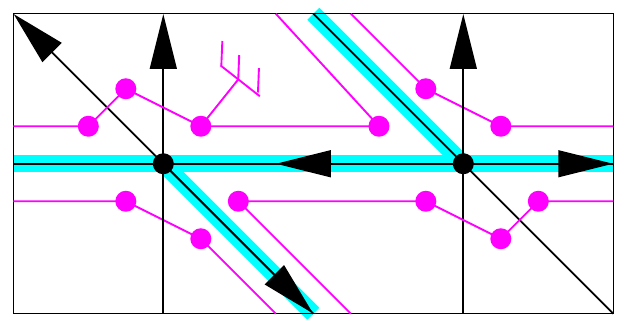} \\
HTC orientation \ \ &\ \ non-HTC orientation \\
\end{tabular}
\caption{A graph that can be represented by two different unicellular
  maps.}
\label{fig:twounicellular}
\end{figure}

For that purpose we use the following definition of $\gamma$
introduced in~\cite{GKL15}.  Consider a particular orientation of
$G$. Let $C$ be a cycle that is given with an arbitrary direction ($C$ is
not necessarily a directed cycle). Then $\gamma(C)$ is defined by:
$$\gamma (C) \ \ =\ \ \#\ \text{edges leaving $C$ on its right} \ \ -\  \
\#\ \text{edges leaving $C$ on its left}$$

By the Schnyder property, it is clear that in a toroidal Schnyder
wood, a monochromatic cycle $C$ always satisfies $\gamma(C)=0$.
Consider a crossing Schnyder wood of $G$ and $C_1,C_2$ two
monochromatic cycles of different colors. Thus we have
$\gamma(C_{1})= \gamma(C_{2})= 0$.  By \cite[Theorem~7]{GL13}, the two
cycles $C_1,C_2$ are non-contractible and non-homologous, thus they
form a basis for the homology.  While returning a $0$-homologous
oriented subgraph, the value of $\gamma$ on a given cycle does not
change. Thus any HTC Schnyder wood also satisfies
$\gamma(C_{1})= \gamma(C_{2})= 0$.  Moreover it is proved
in~\cite[Lemma~18]{GKL15} that if a $3$-orientation of a toroidal
triangulation satisfies $\gamma$ equals $0$ for two cycles
forming a basis for the homology, then $\gamma$ equals $0$ for any
non-contractible cycle. Thus any HTC Schnyder wood satisfies $\gamma$
equals $0$ for any non-contractible cycle. We call this property the
$\gamma_0$ property. Note that, for a $3$-orientation, it is
sufficient to satisfy $\gamma$ equals $0$ on any two cycles forming
a basis for the homology to have the $\gamma_0$ property.

Actually the $\gamma_0$ property characterizes the $3$-orientations
that are HTC Schnyder woods. Indeed a consequence of~\cite[Theorem~5
and Lemma~18]{GKL15} is that if two $3$-orientations both satisfy the
\emph{$\gamma_0$} property, then they are homologous to each other and
thus HTC.  Note that for the $3$-orientation on the right of
Figure~\ref{fig:twounicellular}, we have $\gamma$ equals $\pm 2$ for
the horizontal cycle and this explain why this orientation is not HTC
(one can find similar arguments for previous examples of non-HTC
Schnyder woods presented in this paper, see
Figures~\ref{fig:noncrossing} and~\ref{fig:gamma0glue}).

Let us translate this $\gamma_0$ property on $\mathcal U_{r}(n)$.
Consider an element $U$ of $\mathcal U_{r}(n)$ whose edges and stems
are oriented w.r.t.~the root angle as follows: the stems are all
outgoing, and while walking \cw around the unique face of $U$ from
$a_0$, the first time an edge is met, it is oriented \ccw w.r.t.~ the
face of $U$. Then one can compute $\gamma$ on the cycles of $U$ (edges
and stems count).  We say that an unicellular map of
$\mathcal U_{r}(n)$ satisfies the $\gamma_0$ property if $\gamma$
equals zero on its (non-contractible) cycles. Let us call
$\mathcal U_{r,b,\gamma_0}(n)$ the set of elements of
$\mathcal U_{r,b}(n)$ satisfying the $\gamma_0$ property. So the
output of \aps given by Theorem~\ref{th:uni} is an element of
$\mathcal U_{r,b,\gamma_0}(n)$.

Let $\mc T_r(n)$ be the set of toroidal triangulations on $n$ vertices
rooted at an angle that is not in the \cw interior of a separating
triangle. Then we
have the following bijection:

 \begin{theorem}
\label{th:bij1}
There is a bijection between $\mc T_r(n)$ and
$\mathcal U_{r,b,\gamma_0}(n)$.
\end{theorem}

\begin{proof}
  Consider the mapping $g$ that associates to an element of
  $\mc T_r(n)$, the output of \aps executed on the minimal HTC
  Schnyder wood w.r.t.~the root face. By the above discussion the
  image of $g$ is in $\mathcal U_{r,b,\gamma_0}(n)$ and $g$ is
  injective since one can recover the original triangulation from its
  image by Theorem~\ref{th:recover}.

  Conversely, given an element $U$ of $\mathcal U_{r,b,\gamma_0}(n)$
  with root angle $a_0$, one can build a toroidal map $G$ by the
  complete closure procedure described in Section~\ref{sec:close}. The
  number of stems and edges of $U$ implies that $G$ is a
  triangulation. Recall that $a_0$ defines an orientation on the edges
  and stems of $U$.  Consider the orientation $D$ of $G$ induced by
  this orientation. Since $U$ is balanced, the execution of \aps on
  $(G,a_0)$ corresponds to the cycle in the angle graph of $U$
  obtained by starting from the root angle and walking \cw in the face
  of $U$. Thus the output of \aps executed on $(G,a_0)$ is $U$.  It
  remains to show that $G$ is appropriately rooted and that $D$
  corresponds to the minimal HTC Schnyder wood w.r.t.~ this root, then
  $g$ will be surjective.

  First note that by definition of $\mathcal U_{r}(n)$, the
  orientation $D$ is a $3$-orientation.

  Suppose by contradiction that $a_0$ is in the strict interior of a
  separating triangle. Then, since we are considering a
  $3$-orientation, by Euler's formula, the edges in the interior of
  this triangle and incident to its border are all entering the
  border. So \aps started from the strict interior cannot visit the
  vertices on the border of the triangle and outside. Thus the output
  of \aps is not a toroidal unicellular map, a contradiction. So $a_0$
  is not in the strict interior of a separating triangle.

  The $\gamma_0$ property of $U$ implies that $\gamma$ equals zero on
  two cycles of $U$. Hence these two cycles considered in $G$ also
  satisfy $\gamma$ equals $0$ and form a basis for the homology. So
  $D$ is an HTC Schnyder wood.
  
  Suppose by contradiction that $D$ is not minimal. Then, by
  Lemma~\ref{lem:maximal}, it contains a clockwise (non-empty)
  $0$-homologous oriented subgraph w.r.t.~$f_0$. With the notations of
  Section~\ref{sec:schnyderwoods}, let $T$ be such a subgraph with
  $\phi(T)=-\sum_{F\in\mc{F}'}\lambda_F\phi(F)$, with
  $\lambda\in\mathbb{N}^{|\mc{F}'|}$.  Let $\lambda_{F_0}=0$, and
  $\lambda_{\max}=\max_{F\in\mathcal{F}}\lambda_F$.  For
  $0\leq i\leq \lambda_{\max}$, let
  $X_{i}=\{F\in\mathcal{F}\,|\,\lambda_F\geq i\}$.  For
  $1\leq i \leq \lambda_{\max}$, let $T_i$ be the oriented subgraph
  such that $\phi(T_i)=-\sum_{F\in X_i}\phi(F)$.  Then we have
  $\phi(T)=\sum_{1\leq i \leq \lambda_{\max}} \phi(T_i)$.  Since $T$
  is an oriented subgraph, we have
  $\phi(T)\in\{-1,0,1\}^{|E(G)|}$. Thus for any edge of ${G}$,
  incident to faces $F_1$ and $F_2$, we have
  $(\lambda_{F_1}-\lambda_{F_2})\in\{-1,0,1\}$.  So, for
  $1\leq i\leq \lambda_{\max}$, the oriented graph $T_i$ is the
  frontier between the faces with $\lambda$ value equal to $i$ and
  $i-1$.  So all the $T_i$ are edge disjoint and are oriented
  subgraphs of $D$.  Since $T$ is non-empty, we have
  $\lambda_{\max}\geq 1$, and $T_1$ is non-empty.  All the edges of
  $T_1$ have a face of $X_1$ on their right and a face of $X_0$ on
  their left.  Since $U$ is an unicellular map, and $T_1$ is a
  (non-empty) $0$-homologous oriented subgraph, at least one edge of
  $T_1$ corresponds to a stem of $U$. Let $s$ be the last stem of $U$
  corresponding to an edge of $T_1$ that is reattached by the complete
  closure procedure. Consider the step where $s$ is reattached.  As
  the root angle (and thus $f_0$) is in the special face (see the
  terminology of Section~\ref{sec:close}), the special face is in the
  region defined by $X_0$. Thus it is on the left of $s$ when it is
  reattached. This contradicts the fact that $U$ is balanced.  Thus
  $D$ is the minimal HTC Schnyder wood w.r.t.~$f_0$.
\end{proof}
\section{The lattice of HTC Schnyder woods}
\label{sec:flip}

In this section, we push further the study of HTC Schnyder woods in
order to remove the root and the balanced property of the unicellular
maps considered in Theorem~\ref{th:bij1} and obtain a simplified
bijection in Theorem~\ref{th:bij2} of Section~\ref{sec:bij2}.

Consider a toroidal triangulation $G$ given with a crossing Schnyder
wood. Let $D_0$ be the corresponding $3$-orientation of $G$. Let $f_0$
be any face of $G$. Recall from Section~\ref{sec:schnyderwoods} that
$O(G)$ denotes the set of all the orientations of $G$ that are
homologous to $D_0$. The elements of $O(G)$ are the HTC Schnyder woods
of $G$ and $(O(G),\leq_{f_0})$ is a distributive lattice.

We need to reduce the graph ${G}$. We call an edge of ${G}$
\emph{rigid} w.r.t.~$O(G)$ if it has the same orientation in all the
elements of $O(G)$. Rigid edges do not play a role for the structure
of $O(G)$. We delete them from ${G}$ and call the obtained embedded
graph $\widetilde{G}$. Note that this graph is embedded but it is not
necessarily a map, as some faces may not be homeomorphic to open disks.
Note also that $\widetilde{G}$ might be empty
if all the edges are rigid, i.e. $|O(G)|=1$ and $\widetilde{G}$ has no
edge but a unique face that is all the surface.

\begin{lemma}
\label{lem:non-rigid}
Given an edge $e$ of $G$, the following are equivalent:
\begin{enumerate}
\item $e$ is non-rigid
\item  $e$ is contained in a
$0$-homologous oriented subgraph of $D_0$
\item  $e$ is contained in a
$0$-homologous oriented subgraph of any element of $O(G)$
\end{enumerate}
\end{lemma}

\begin{proof}
  $(1 \Longrightarrow 3)$ Let $D\in O(G)$. If $e$ is non-rigid,
  then it has a different orientation in two elements $D',D''$ of
  $O(G)$.  Then we can assume by symmetry that $e$ has a different
  orientation in $D$ and $D'$ (otherwise in $D$ and $D''$ by
  symmetry). Since $D,D'$ are homologous to $D_0$, they are also
  homologous to each other. So $T=D\setminus D'$ is a $0$-homologous
  oriented subgraph of $D$ that contains $e$.

 $(3 \Longrightarrow 2)$ Trivial since $D_0\in O(G)$

  $(2 \Longrightarrow 1)$ If an edge $e$ is contained in a $0$-homologous
  oriented subgraph $T$ of $D_0$. Then let $D$ be the element of
  $O(G)$ such that $T=D_0\setminus D$. Clearly $e$ is oriented
  differently in $D$ and $D_0$, thus it is non-rigid.
\end{proof}

By Lemma~\ref{lem:non-rigid}, one can build $\widetilde{G}$ by keeping
only the edges that are contained in a $0$-homologous oriented
subgraph of $D_0$.  Note that this implies that all the edges of
$\widetilde{G}$ are incident to two distinct faces of $\widetilde{G}$.
Denote by $\widetilde{\mathcal{F}}$ the set of oriented subgraphs of
$\widetilde{G}$ corresponding to the boundaries of faces of
$\widetilde{G}$ considered counterclockwise.  Let $\widetilde{f}_0$ be
the face of $\widetilde{G}$ containing $f_0$ and $\widetilde{F}_0$ be
the element of $\widetilde{\mathcal{F}}$ corresponding to the boundary
of $\widetilde{f}_0$. Let
$\widetilde{\mathcal{F}}'=\widetilde{\mathcal{F}}\setminus
\widetilde{F}_0$.
The proof of~\cite[Theorem~7]{GKL15} shows that the elements of
$\widetilde{\mathcal{F}}'$ are sufficient to generate the entire
lattice $(O(G),\leq_{f_0})$, i.e. two elements $D,D'$ of $O(G)$ are
linked in the Hasse diagram of the lattice, with $D\leq_{f_0}D'$, if
and only if $D\setminus D' \in \widetilde{\mathcal{F}}'$.

\begin{lemma}
\label{lem:necessary}
For every element $ \widetilde{F}\in \widetilde{\mathcal{F}}$ there
exists $D$ in $O(G)$ such that $\widetilde{F}$ is an oriented subgraph
of $D$. 
\end{lemma}
\begin{proof}
  Let $ \widetilde{F}\in \widetilde{\mathcal{F}}$. Let $D$ be an
  element of $O(G)$ that maximize the number of edges of
  $\widetilde{F}$ that have the same orientation in $\widetilde{F}$
  and $D$ (i.e. that maximize the number of edges of $D$ oriented
  counterclockwise on the border of the face of $\widetilde{G}$
  corresponding to $\widetilde{F}$).  Suppose by contradiction that
  there is an edge $e$ of $\widetilde{F}$ that does not have the same
  orientation in $\widetilde{F}$ and $D$. Edge $e$ is in
  $\widetilde{G}$ so it is non-rigid.  Let $D'\in O(G)$ such that $e$
  is oriented differently in $D$ and $D'$. Let $T=D\setminus D'$.
  By~\cite[Claim~1 of the proof of Theorem~7]{GKL15}, there exists
  edge-disjoint oriented subgraphs $T_1,\ldots,T_k$ of $D$ such that
  $\phi(T)=\sum_{1\leq i \leq k} \phi(T_i)$, and, for
  $1\leq i \leq k$, there exists
  $\widetilde{X}_i\subseteq \widetilde{{\mathcal{F}}}'$ and
  $\epsilon_i\in\{-1,1\}$ such that
  $\phi(T_i)=\epsilon_i\sum_{\widetilde{F}'\in \widetilde{X}_i}
  \phi(\widetilde{F}')$.
  W.l.o.g., we can assume that $e$ is an edge of $T_1$.  Let $D''$ be
  the element of $O(G)$ such that $T_1=D\setminus D''$.  The oriented
  subgraph $T_1$ intersects $\widetilde{F}$ only on edges of $D$
  oriented clockwise on the border of $\widetilde{F}$. So $D''$
  contains strictly more edges oriented counterclockwise on the border
  of the face $\widetilde{F}$ than $D$, a contradiction.  So all the
  edges of $\widetilde{F}$ have the same orientation in $D$.  So
  $\widetilde{F}$ is a $0$-homologous oriented subgraph of $D$.
\end{proof}

By Lemma~\ref{lem:necessary}, for every element
$ \widetilde{F}\in \widetilde{\mathcal{F}}'$ there exists $D$ in
$O(G)$ such that $\widetilde{F}$ is an oriented subgraph of $D$. Thus
there exists $D'$ such that $\widetilde{F}=D\setminus D'$ and $D,D'$
are linked in the Hasse diagram of the lattice.  Thus the elements of
$\widetilde{\mathcal{F}}'$ form a minimal set that generates the
lattice.

Let $D_{\max}$ (resp. $D_{\min}$) be the maximal (resp. minimal)
element of $(O(G),\leq_{f_0})$. 

\begin{lemma}
\label{lem:maxtilde}
  $\widetilde{F}_0$ (resp. $-\widetilde{F}_0$) is an oriented subgraph
  of $D_{\max}$ (resp. $D_{\min}$). 
\end{lemma}

\begin{proof}
  By Lemma~\ref{lem:necessary}, there exists $D$ in $O(G)$ such that
  $\widetilde{F}$ is an oriented subgraph of $D$. Let
  $T=D\setminus D_{\max}$. Since $D\leq_{f_0} D_{\max}$, the proof
  of~\cite[Theorem~7]{GKL15}, shows that the characteristic flow of
  $T$ can be written as a combination with positive coefficients of
  characteristic flows of $\widetilde{\mathcal{F}}'$, i.e.
  $\phi(T)=\sum_{\widetilde{F}\in
    \widetilde{\mathcal{F}}'}\lambda_F\phi(\widetilde{F})$
  with $\lambda\in\mathbb{N}^{|\mc{F}'|}$. So $T$ is disjoint from
  $\widetilde{F}_0$.  Thus $\widetilde{F}_0$ is an oriented subgraph
  of $D_{\max}$. The proof is similar for $D_{\min}$.
  \end{proof}

  Note that the above three lemmas hold in a more general context than
  just $O(G)$. Actually they hold for any lattice of homologous
  orientations on an oriented surface (see~\cite{LevHDR}).  From now
  on we use some specific properties of the object considered in this
  paper, i.e. HTC Schnyder woods.

\begin{lemma}
\label{lem:facialwalktilde}
Consider an orientation $D$ in $O(G)$ and a closed walk $W$ of
$\widetilde{G}$. If on the left (resp. right) side of $W$, there is no
incident edges of $\widetilde{G}$, and no outgoing incident edges of
$D$, then $W$ is a contractible triangle with its contractible region
on its left (resp. right) side.
\end{lemma}

\begin{proof} 
  Consider a closed walk $W$ of $\widetilde{G}$ such that on its left
  side there is no incident edge of $\widetilde{G}$, and no outgoing
  incident edges of $D$. Let $k$ be the length of $W$.  Let $W_{left}$
  be the edges of $D$ that are incident to the left side of $W$. By
  assumption they are all entering $W$. Note that $W$ cannot cross
  itself otherwise it has at least one incident edge of
  $\widetilde{G}$ on its left side. However it may have repeated
  vertices but in that case it intersects itself tangencially on the
  right side.

  Suppose first that $W$ is non-contractible.  Then consider the
  closed walk $W^*$ of the dual orientation $D^*$ that is obtained by
  considering all the dual edges of $W_{left}$ with their
  corresponding orientation. Since all the edges of $W_{left}$ are
  entering $W$ we have that $W^*$ is an oriented closed walk. Moreover
  it is non-contractible and thus contains an oriented
  non-contractible cycle, a contradiction to
  Lemma~\ref{lem:incomingedges}.  So $W$ is contractible.  Since $W$
  can intersect itself only tangencially on the right side, the
   region delimited by $W$ and located on its left side is
  connected.

  Suppose that $W$ has its contractible region on its left side.
  Consider the graph $G'$ obtained from $G$ by keeping only the
  vertices and edges that lie in the contractible region delimited by
  $W$, including $W$. The vertices of $W$ appearing several times are
  duplicated so that $G'$ is a plane triangulation of a $k$-cycle.
  Let $n',m',f'$ be the number of vertices, edges and faces of $G'$.
  By Euler's formula, $n'-m'+f'=2$. All the inner faces have size $3$
  and the outer face has size $k$, so $2m'=3(f'-1)+k$.  All the inner
  vertices have outdegree $3$ as we are considering a $3$-orientation
  of $G$. All the edges of $W_{left}$ are oriented toward $W$, and
  there are $k$ outer edges, so $m'=3(n'-k)+k$. Combining these three
  equalities gives $k=3$, i.e. $W$ is a triangle and the lemma holds.

  Suppose now that $W$ has its contractible region on its right
  side. Then similarly as above, consider the graph $G'$ obtained from
  $G$ by keeping all the vertices and edges that lie in the
  contractible region delimited by $W$, including $W$. This time the
  vertices of $W$ appearing several times are not duplicated. Since
  $W$ can intersect itself only tangencially on the right side, we
  have that $G'$ is a plane map whose outer face boundary is $W$ and
  whose interior is triangulated.  As above, let $n',m',f'$ be the
  number of vertices, edges and faces of $G'$.  By Euler's formula,
  $n'-m'+f'=2$. All the inner faces have size $3$ and the outer face
  has size $k$, so $2m'=3(f'-1)+k$. Since there is no outgoing
  incident edges of $D$ on the left side of $W$, all the vertices of
  $G'$ have outdegree $3$ and $m'=3n'$.  Combining these three
  equalities gives $k=-3$, a contradiction.
\end{proof}

The boundary of a face of $\widetilde{G}$ may be composed of several
closed walks. Let us call \emph{quasi-contractible} the faces of
$\widetilde{G}$ that are homeomorphic to a disk or to a disk with
punctures.  Note that such a face may have several boundaries (if
there is some punctures) and then the face is not contractible, but
exactly one of these boundaries contains all the other in its
contractible region. Let us call \emph{outer facial walk} this special
boundary. Then we have the following:

\begin{lemma}
\label{lem:contractibletilde}
All the faces of $\widetilde{G}$ are quasi-contractible and their
outer facial walk is a (contractible) triangle.
\end{lemma}

\begin{proof}
  Suppose by contradiction that there is a face $\widetilde{f}$ of
  $\widetilde{G}$ that is not quasi-contractible or whose outer facial
  walk is not a contractible triangle. Let $\widetilde{F}$ be the
  element of $\widetilde{\mathcal{F}}$ corresponding to the boundary
  of $\widetilde{f}$. By Lemma~\ref{lem:necessary}, there exists an
  orientation $D$ in $O(G)$ such that $\widetilde{F}$ is an oriented
  subgraph of $D$.

  All the faces of $G$, are contractible triangles. Thus
  $\widetilde{f}$ is not a face of $G$ and contains in its interior at
  least one edge of $G$.  Start from any such edge $e$ and consider
  the \emph{left-walk} $W=(e_i)_{i\geq 0}$ of $D$ obtained by the
  following: if the edge $e_i$ is entering a vertex $v$, then
  $e_{i+1}$ is choosen among the three edges leaving $v$ as the edge
  that is on the left coming from $e_i$ (i.e. the first one while
  going \cw around $v$).  Suppose that for $i\geq 0$, edge $e_i$
  is entering a vertex $v$ that is on the border of $\widetilde{f}$.
  Recall that by definition $\widetilde{F}$ is oriented \ccw according
  to its interior, so either $e_{i+1}$ is in the interior of
  $\widetilde{f}$ or $e_{i+1}$ is on the border of
  $\widetilde{f}$. Thus $W$ cannot leave $\widetilde{f}$.

  Since $G$ has a finite number of edges, some edges are used several
  times in $W$.  Consider a minimal subsequence
  $W'=e_k, \ldots, e_\ell$ such that no edge appears twice and
  $e_k=e_{\ell+1}$.  Thus $W$ ends periodically on the sequence of
  edges $e_k, \ldots, e_\ell$. By Lemma~\ref{lem:facialwalktilde}, all
  the closed walks that are part of $\widetilde{F}$ have some outgoing
  incident edges of $D$ on their left side. Thus we have that $W'$
  contains at least one edge that is not an edge of $\widetilde{F}$,
  thus it contains at least one rigid edge.

  By construction, all the edges on the left side of $W'$ are
  entering.  Suppose that $W'$ is not contractible. Then the oriented
  closed walk of the dual orientation $D^*$ that is obtained by
  considering all the dual edges of its incident edges on the left
  side gives a contradiction to Lemma~\ref{lem:incomingedges}. So $W'$
  is contractible. So it is a $0$-homologous oriented subgraph of $D$,
  thus all its edges are non-rigid by Lemma~\ref{lem:non-rigid}, a
  contradiction.
\end{proof}

By Lemma~\ref{lem:contractibletilde}, every face of $\widetilde{G}$ is
quasi-contractible and its outer facial walk is a contractible
triangle.  So $\widetilde{G}$ contains all the contractible triangles
of $G$ whose interiors are maximal by inclusion, i.e. it contains all
the edges that are not in the interior of a separating triangle. In
particular, $\widetilde{G}$ is non-empty and $|O(G)|\geq 2$.  The
status (rigid or not) of an edge lying inside a separating triangle is
determined as in the planar case: such an edge is rigid if and only if
it is in the interior of a separating triangle and incident to this
triangle. Thus an edge of $G$ is rigid if and only if it is in the
interior of a separating triangle and incident to this triangle.

Since $(O(G),\leq_{f_0})$ is a distributive lattice, any element $D$
of $O(G)$ that is distinct from $D_{\max}$ and $D_{\min}$ contains at
least one neighbor above and at least one neighbor below in the Hasse
diagram of the lattice. Thus it has at least one face of
$\widetilde{G}$ oriented \ccw and at least one face of $\widetilde{G}$
oriented \cww. Thus by Lemma~\ref{lem:contractibletilde}, it contains
at least one contractible triangle oriented \ccw and at least one
contractible triangle oriented \cww. Next lemma shows that this
property is also true for $D_{\max}$ and $D_{\min}$.

\begin{lemma}
\label{lem:trianglef0}
  In $D_{\max}$ (resp. $D_{\min}$) there is a \ccw (resp. \cww)
  contractible triangle containing $f_0$, and a \cw (resp. \ccww)
  contractible triangle not containing $f_0$.
\end{lemma}

\begin{proof}
  By Lemma~\ref{lem:contractibletilde}, $\widetilde{f}_0$ is
  quasi-contractible and its outer facial walk is a contractible
  triangle $T$.  By lemma~\ref{lem:maxtilde}, $\widetilde{F}_0$ is an
  oriented subgraph of $D_{\max}$.  Thus $T$ is oriented \ccw and
  contains $f_0$.  The second part of the lemma is clear since
  $|O(G)|\geq 2$ so $D_{\max}$ has at least one neighbor below in the
  Hasse diagram of the lattice. Similarly for $D_{\min}$.
\end{proof}

Thus by above remarks and Lemma~\ref{lem:trianglef0}, all the HTC
Schnyder woods have at least one triangle oriented \ccw and at least
one triangle oriented \cww. Note that this property does not
characterize HTC Schnyder woods.  Figure~\ref{fig:gamma0glue} gives an
example of a Schnyder wood that is not HTC but satisfies the
property. Note also that not all Schnyder woods satisfy the
property. The right of Figure~\ref{fig:noncrossing} is an example of a
Schnyder wood that is no HTC and has no oriented contractible
triangle.

Lemma~\ref{lem:trianglef0} is used in the next section to obtained a
bijection with unrooted unicellular maps.

\section{Bijection with unrooted unicellular maps}
\label{sec:bij2}

To remove the root and the balanced property of the unicellular maps
considered in Theorem~\ref{th:bij1}, we have to root the toroidal
triangulation more precisely than before. We say that an angle is not
\emph{in the \cw interior of a separating triangle} if it is not in
its contractible region, or if it is incident to a vertex $v$ of the
triangle and situated just before an edge of the triangle in \ccw
order around $v$ (see Figure~\ref{fig:rootangles}).

\begin{figure}[h!]
\center
\includegraphics[scale=0.5]{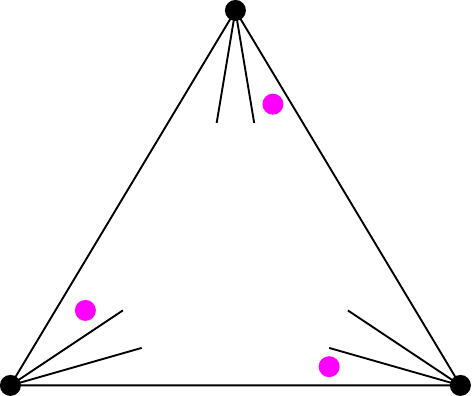}
\caption{Angles that are in a separating triangle but not in its \cw
  interior.}
\label{fig:rootangles}
\end{figure}

Consider a toroidal triangulation $G$. Consider a root angle $a_0$
that is not in the \cw interior of a separating triangle.  Note that
the choice of $a_0$ is equivalent to the choice of a root vertex $v_0$
and a root edge $e_0$ incident to $v_0$ such that none is in the
interior of a separating triangle.  Consider the orientation of the
edges of $G$ corresponding to the minimal HTC Schnyder wood w.r.t.~the
root face $f_0$.  By Lemma~\ref{lem:trianglef0}, there is a \cw
triangle containing $f_0$. Thus by the choice of $a_0$, the edge $e_0$
is leaving the root vertex $v_0$. This is the essential property used
in this section.  Consider the output $U$ of \aps on $(G,a_0)$.  Since
$e_0$ is leaving $v_0$ and $a_0$ is just before $e_0$ in \ccw order
around $v_0$, the execution of \aps starts by Case 2 and $e_0$
corresponds in $U$ to a stem $s_0$ attached to $v_0$. We call this
stem $s_0$ the \emph{root stem}.

The recovering method defined in Theorem~\ref{th:recover} says that
$s_0$ is the last stem reattached by the procedure. So there exists a
sequence of admissible triples of $U$ (see the terminology and
notations of Section~\ref{sec:close}) such that $s_0$ belongs to the
last admissible triple. Let $U_0=U$ and for $1\leq k \leq 2n-2$, the
map $U_{k}$ is obtained from $U_{k-1}$ by closing any admissible
triple that does no contain $s_0$. As noted in
Section~\ref{sec:close}, the special face of $U_{2n-2}$ is a
quadrangle with exactly one stem. This stem being $s_0$, we are in the
situation of Figure~\ref{fig:rootstem}.

\begin{figure}[h!]
\center
\includegraphics[scale=0.4]{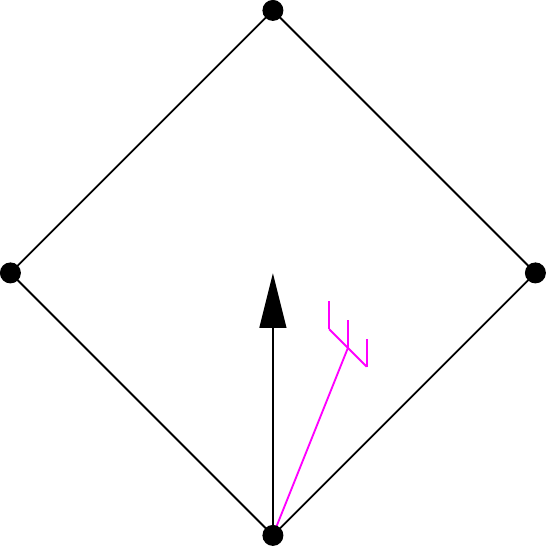} 
\caption{The situation just before the last stem (i.e. the root stem)
  is reattached}
\label{fig:rootstem}
\end{figure}

Consequently, if one removes the root stem $s_0$ from $U$ to obtain an
unicellular map $U'$ with $n$ vertices, $n+1$ edges and $2n-2$ stems,
one can recover the graph $U_{2n-2}$ by applying a complete closure
procedure on $U'$ (see example of Figure~\ref{fig:uprime}). Note that
then, there are four different ways to finish the closure of
$U_{2n-2}$ to obtain an oriented toroidal triangulation. This four
cases corresponds to the four ways to place the (removed) root stem in
a quadrangle, they are obtained by pivoting Figure~\ref{fig:rootstem}
by 0°, 90°, 180° and 270°. Note that only one of this four cases leads
to the original rooted triangulation $G$, except if there are some
symmetries (like in the example of Figure~\ref{fig:uprime}).

\begin{figure}[h!]
\center
\begin{tabular}{cc}
\includegraphics[scale=0.5]{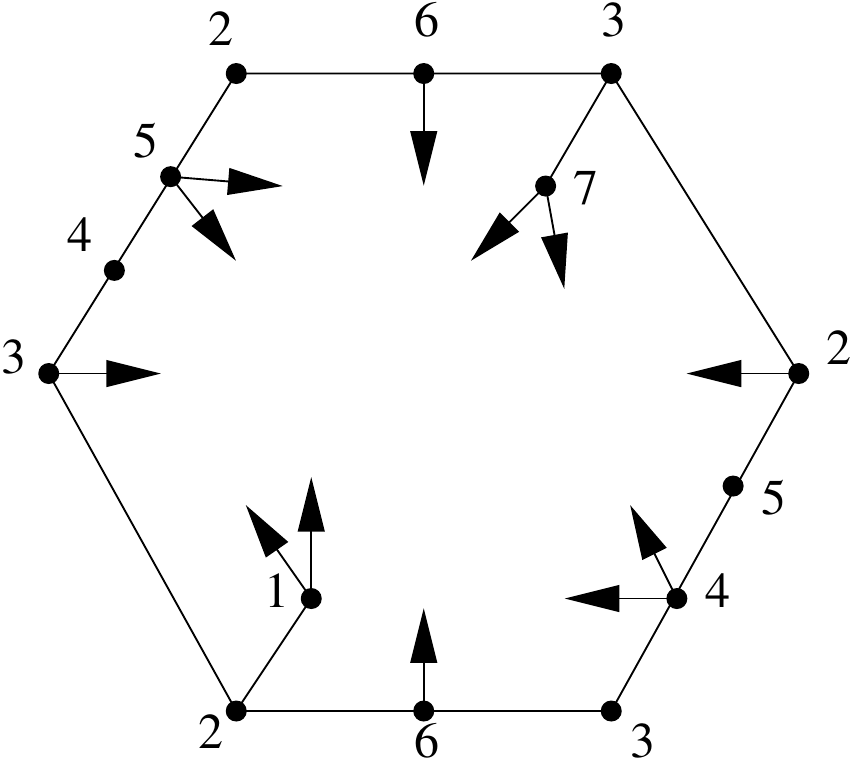}
 & 
\includegraphics[scale=0.5]{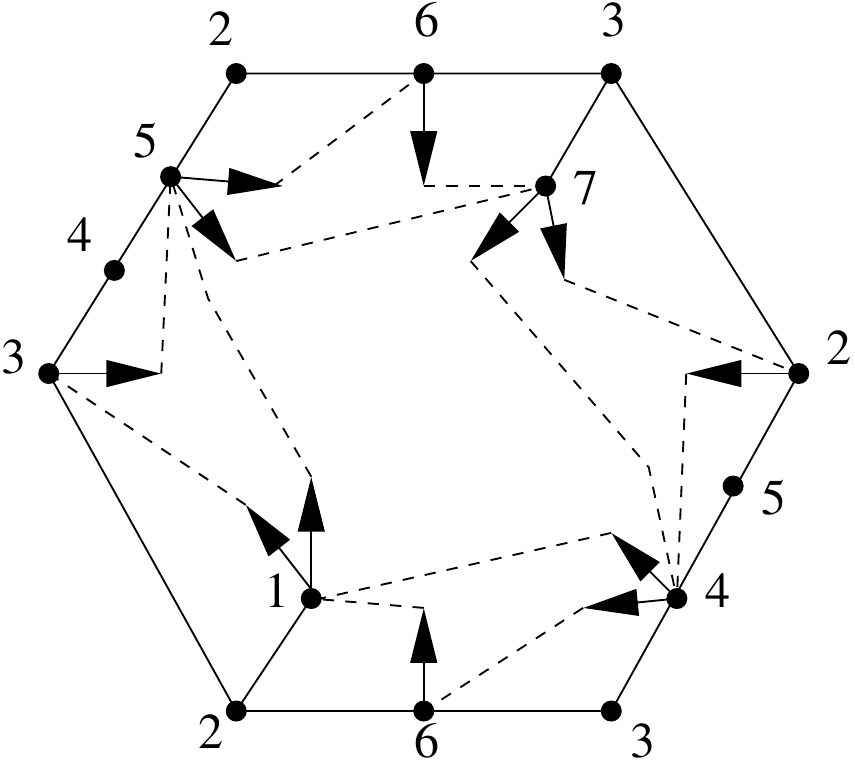}\\
Non-rooted output & Partial recovering \\
\end{tabular}
\caption{Example of $K_7$ where the root angle, the root stem and the
  orientation w.r.t.~the root angle have been removed from the output
  of Figure~\ref{fig:tore-example}. The complete closure procedure
  leads to a quadrangular face.}
\label{fig:uprime}
\end{figure}

Let $\mathcal U(n)$ denote the set of (non-rooted) toroidal
unicellular maps, with exactly $n$ vertices, $n+1$ edges and $2n-2$
stems satisfying the following: every vertex has exactly $2$ stems,
except if the map is hexagonal, the two corners having exactly $1$
stem each, and if the map is square, the only corner having no stem at
all. Note that the output of Theorem~\ref{th:uni} on an appropriately
rooted toroidal triangulation is an element of $\mathcal U(n)$ when
the root stem is removed.

Note that an element $U'$ of $\mathcal U(n)$ is non-rooted so we
cannot orient automatically its edges w.r.t.~ the root angle like in
Section~\ref{sec:bijection}.  Nevertheless one can still orient all
the stems as outgoing and compute $\gamma$ on the cycles of $U'$ by
considering only its stems in the counting (and not the edges nor the
root stem anymore).  We say that an unicellular map of $\mathcal U(n)$
satisfies the $\gamma_0$ property if $\gamma$ equals zero on its
(non-contractible) cycles.  Let us call $\mathcal U_{\gamma_0}(n)$ the
set of elements of $\mathcal U(n)$ satisfying the $\gamma_0$ property.

A surprising property is that an element $U'$ of $\mathcal U(n)$
satisfies the $\gamma_0$ property if and only if any element $U$ of
$\mathcal U_{r}(n)$ obtained from $U'$ by adding a root stem anywhere
in $U'$ satisfies the $\gamma_0$ property (note that in $U$ we count
the edges and the root stem to compute $\gamma$).  One can see this by
considering the unicellular map of
Figure~\ref{fig:hexasquaregamma}. It represents the general case of
the underlying rooted hexagon of $U$. The edges represent in fact
paths (some of which can be of length zero).  One can check that it
satisfies $\gamma$ equals zero on its (non-contractible) cycles.  It
corresponds exactly to the set of edges that are taken into
consideration when computing $\gamma$ on $U$ but not when computing
$\gamma$ on $U'$. Thus it does not affect the counting (the tree-like
parts are not represented since they do not affect the value
$\gamma$). So the output of Theorem~\ref{th:uni} on an appropriately
rooted toroidal triangulation is an element of
$\mathcal U_{\gamma_0}(n)$ when the root stem is removed.

\begin{figure}[!h]
\center
\includegraphics[scale=0.5]{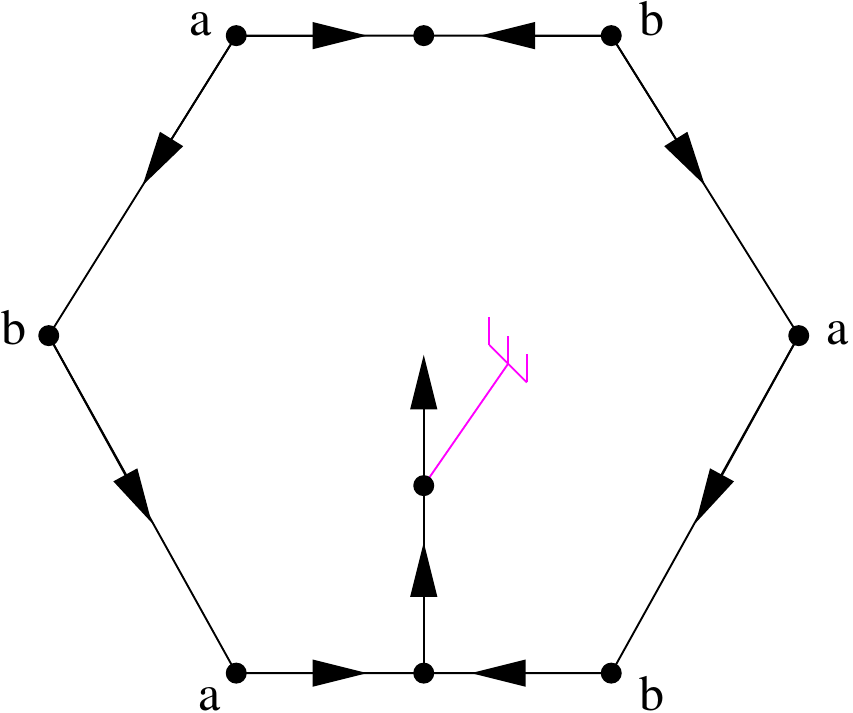}
\caption{The parts of the unicellular map showing the correspondence
  while computing $\gamma$ with or without the orientation w.r.t.~
  the root plus the root stem.}
\label{fig:hexasquaregamma}
\end{figure}

For the particular case of $K_7$, the difference between the rooted
output of Figure~\ref{fig:tore-example} and the non-rooted output of
Figure~\ref{fig:uprime} is represented on Figure~\ref{fig:uprimediff}
(one can superimpose the last two to obtain the first). One can check
that these three unicellular maps (rooted, non-rooted and the
difference) all satisfy $\gamma$ equals zero on their cycles.

\begin{figure}[h!]
\center
\includegraphics[scale=0.5]{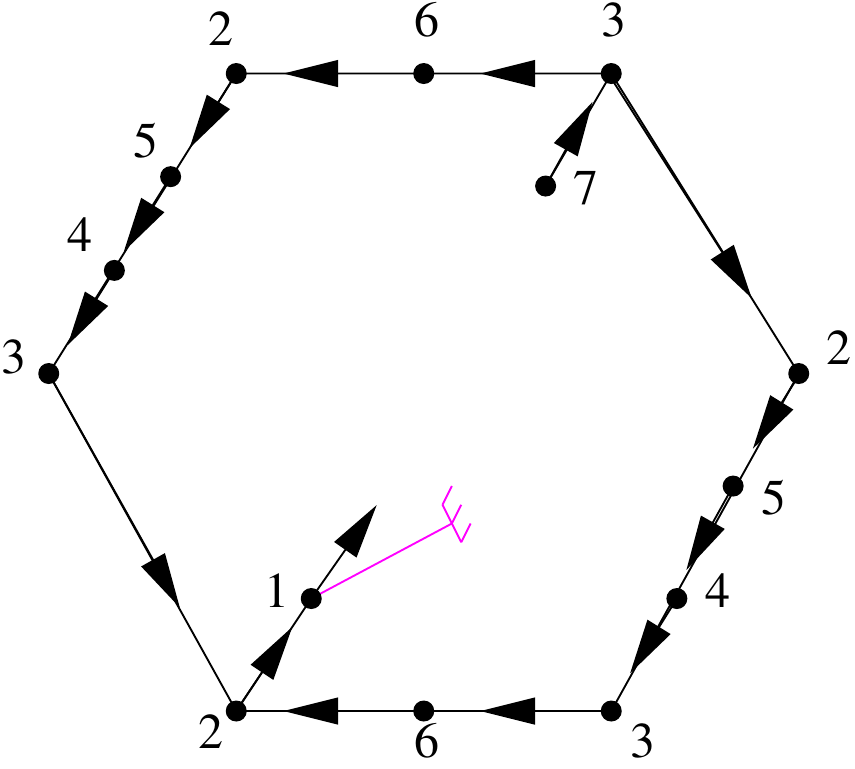}
\caption{The difference between the rooted output of
  Figure~\ref{fig:tore-example} and the non-rooted output of
  Figure~\ref{fig:uprime}.}
\label{fig:uprimediff}
\end{figure}

There is an ``almost'' four-to-one correspondence between toroidal
triangulations on $n$ vertices, given with a root angle that is not in
the \cw interior of a separating triangle, and elements of
$\mathcal U_{\gamma_0}(n)$. The ``almost'' means that if the
automorphism group of an element $U$ of $\mathcal U_{\gamma_0}(n)$ is
not trivial, some of the four ways to add a root stem in $U$ are
isomorphic and lead to the same rooted triangulation. In the example
of Figure~\ref{fig:uprime}, one can root in four ways the quadrangle
but this gives only two different rooted triangulations (because of
the symmetries of $K_7$). We face this problem by defining another
class for which we can formulate a bijection.

Let $\mc T(n)$ be the set of toroidal maps on $n$ vertices, where all
the faces are triangles, except one that is a quadrangle and which is
not in a separating triangle. Then we have the following bijection:

 \begin{theorem}
\label{th:bij2}
There is a bijection between $\mc T(n)$ and
$\mathcal U_{\gamma_0}(n)$.
\end{theorem}

\begin{proof}
  Let $a$ (for ``add'') be an arbitrarily chosen mapping defined on
  the maps $G'$ of $\mc T(n)$ that adds a diagonal $e_0$ in the
  quadrangle of $G'$ and roots the obtained toroidal triangulation $G$
  at a vertex $v_0$ incident to $e_0$ (this defines the root angle $a_0$
  situated just before $e_0$ in \ccw order around $v_0$).  Note that
  the added edge cannot create a separating 2-cycle, since otherwise
  the quadrangle would be in a separating triangle. Moreover the root
  angle of $G$ is not in the \cw interior of a separating
  triangle. Thus the image of $a$ is in $\mc T'_r(n)$, the subset of
  $\mc T_r(n)$ corresponding to toroidal triangulations rooted at an
  angle that is not in the \cw interior of a separating triangle.

  Let $\mathcal U'_{r,b,\gamma_0}(n)$ be the elements of
  $\mathcal U_{r,b,\gamma_0}(n)$ that have their root angle just
  before a stem in \ccw order around the root vertex.  Consider the
  mapping $g$, defined in the proof of Theorem~\ref{sec:bijection}.
  By above remarks and Theorem~\ref{sec:bijection}, the image of $g$
  restricted to $\mc T'_r(n)$ is in $\mathcal U'_{r,b,\gamma_0}(n)$.
  Let $r$ (for ``remove'') be the mapping that associates to an
  element of $\mathcal U'_{r,b,\gamma_0}(n)$ an element of
  $\mathcal U_{\gamma_0}(n)$ obtained by removing the root angle and
  its corresponding stem.  Finally, let $h=r\circ g\circ a$ which
  associates to an element of $\mc T(n)$ an element of
  $\mathcal U_{\gamma_0}(n)$.  Let us show that $h$ is a bijection.

  Consider an element $G'$ of $\mc T(n)$ and its image $U'$ by
  $h$. The complete closure procedure on $U'$ gives $G'$ thus the
  mapping $h$ is injective.

  Conversely, consider an element $U'$ of $\mathcal U_{\gamma_0}(n)$.
  Apply the complete closure procedure on $U'$. At the end of this
  procedure, the special face is a quadrangle whose angles are denoted
  $\alpha^1, \ldots, \alpha^4$. We denote also by
  $\alpha^1,\ldots, \alpha^4$ the corresponding angles of $U'$.  For
  $i\in\{1,\ldots,4\}$, let $U^i$ be the element of
  $\mathcal U_{r}(n)$ obtained by adding a root stem and a root angle
  in the angle $\alpha^i$ of $U'$, with the root angle just before the
  stem in \ccw order around the root vertex.  Note that by the choice
  of $\alpha^i$, the $U^i$ are all balanced.  By above remarks they
  also satisfy the $\gamma_0$ property and thus they are in
  $\mathcal U'_{r,b,\gamma_0}(n)$.

  By the proof of Theorem~\ref{th:bij1}, the complete closure
  procedure on $U^i$ gives a triangulation $G^i$ of $\mc T_r(n)$ that
  is rooted from an angle $a_0^i$ not in the strict interior of a
  separating triangle and oriented according to the minimal HTC
  Schnyder wood w.r.t.~the root face. Moreover the output of \aps
  applied on $(G^i,a_0^i)$ is $U^i$. Since in $U^i$, the root stem is
  present just after the root angle, the first edge seen by the
  execution of \aps on $(G^i,a_0^i)$ is outgoing. So $a_0$ is not in
  the \cw interior of a separating triangle (in a $3$-orientation, all
  the edges that are in the interior of a separating triangle and
  incident to the triangle are entering the triangle). Thus the $G^i$
  are appropriately rooted and are elements of $\mc T'_r(n)$. Removing
  the root edge of any $G_i$, gives the same map $G'$ of $\mc T(n)$.
  Exactly one of the $G_i$ is the image of $G'$ by the mapping
  $a$. Thus the image of $G'$ by $h$ is $U'$ and the mapping $h$ is
  surjective.
\end{proof}

A nice aspect of Theorem~\ref{th:bij2} comparing to
Theorem~\ref{th:bij1} is that the unicellular maps that are considered
are much simpler. They have no root nor balanced property anymore.  It
would be great to use Theorem~\ref{th:bij2} to count and sample
toroidal triangulations.  The main issue comparing to the planar case
is the $\gamma_0$ property.

\section{Higher genus}
\label{sec:highergenus}

The key lemmas that make the encoding method presented in this paper
work are Lemmas~\ref{lem:maximal} and~\ref{lem:incomingedges}. Note
that Theorem~\ref{th:lattice} is proved in a very general form
in~\cite{GKL15}. So one can consider a minimal element in the lattice
and get the equivalent of Lemma~\ref{lem:maximal}. Things are more
complicated for Lemma~\ref{lem:incomingedges} since the existence of
Schnyder woods in higher genus is only conjectured when $g\geq 2$ and
moreover we have no idea of what would be the generalization of
crossing property and thus HTC Schnyder woods. Nevertheless one can
hope to find orientations satisfying the conclusion of
Lemma~\ref{lem:incomingedges} and thus apply the same encoding method
as here. This is what we discuss below.

Recently, Albar, the second author and Knauer~\cite{AGK14} proved the
following:

\begin{theorem}[\cite{AGK14}]
\label{th:AGK}
A simple triangulation on a genus $g\geq 1$ orientable
surface admits an orientation of its edges such that every vertex has
outdegree at least $3$, and divisible by $3$.
\end{theorem}

Theorem~\ref{th:AGK} is proved for simple triangulation but we believe
it to be true for all triangulations. Moreover we hope for a possible
generalization satisfying the conclusion of
Lemma~\ref{lem:incomingedges}:

\begin{conjecture}
\label{conj:accessibility}
A triangulation on a genus $g\geq 1$ orientable surface admits an
orientation of its edges such that every vertex has outdegree at least
$3$, divisible by $3$, and such that there is no oriented
non-contractible cycle in the dual orientation.
\end{conjecture}

Even if Conjecture~\ref{conj:accessibility} is true, more efforts
should be made to obtain a bijection since there might be several
minimal element satisfying the conjecture and one has to identify a
particular one (like the minimal HTC Schnyder wood in our situation).

\section*{Acknowledgments}
We thank Luca Castelli Aleardi, Nicolas Bonichon, Eric Fusy and
Frédéric Meunier for fruitful discussions about this work.


\begin{thebibliography}{00}

\bibitem{AGK14} B.~Albar, D.~Gon\c{c}alves, K.~Knauer, Orienting
  triangulations, manuscript, 2014, arXiv:1412.4979.

\bibitem{AP13} M. Albenque, D. Poulalhon, Generic method for
  bijections between blossoming trees and planar maps, manuscript,
  2013, arXiv:1305.1312.

\bibitem{Ber07} O. Bernardi, Bijective Counting of Tree-Rooted Maps
  and Shuffles of Parenthesis Systems, \emph{Electronic Journal of
  Combinatorics} 14 (2007) R9.

\bibitem{BGH03} N. Bonichon, C. Gavoille, N. Hanusse, An
  information-theoretic upper bound of planar graphs using
  triangulation, Proc. of the 20th Annual Symposium on Theoretical
  Aspects of Computer Science (STACS 2003), \emph{Lecture Notes in
    Computer Science} 2607 (2003) 499-510.

\bibitem{CFL10} L. Castelli Aleardi, E. Fusy, T. Lewiner, Optimal
  encoding of triangular and quadrangular meshes with fixed topology,
  Proc. of the 22nd Canadian Conference on Computational Geometry
  (CCCG 2010).

\bibitem{DPS13} E. Duchi, D. Poulalhon, G. Schaeffer, Uniform random
  sampling of simple branched coverings of the sphere by itself, 
  Proc. of the Twenty-Fifth Annual ACM-SIAM Symposium on
  Discrete Algorithms, Society for Industrial and Applied
  Mathematics, 2013, 294-304.


\bibitem{Fel04} S. Felsner, Lattice structures from planar graphs,
  {\it Electronic Journal of Combinatorics} 11 (2004) R15.


 \bibitem{FO01} H. de Fraysseix, P. Ossona de Mendez, On topological
    aspects of orientations, {\it Discrete Mathematics} 229 (2001) 57-72.


\bibitem{Fus07} E. Fusy, \emph{Combinatoire des cartes planaires et
applications algorithmiques}, PhD thesis, manuscript, 2007.


  \bibitem{Gib10} P.~Giblin, {\it Graphs, surfaces and homology},
    Cambridge University Press, Cambridge, third edition, 2010.


   \bibitem{GL13} D.~Gon\c{c}alves, B.~L\'ev\^eque, Toroidal maps :
     Schnyder woods, orthogonal surfaces and straight-line
     representations, \emph{Discrete and Computational Geometry} 51
     (2014) 67-131.

   \bibitem{GKL15} D.~Gon\c{c}alves, K.~Knauer, B.~L\'ev\^eque,
     Structure of Schnyder labelings on orientable surfaces,
     manuscript, 2015, arXiv:1501.05475.

   \bibitem{Kan96} G.~Kant, Drawing planar graphs using the canonical
     ordering, \emph{Algorithmica} 16 (1996) 4-32.

\bibitem{LevHDR} B.~L\'ev\^eque, \emph{Generalization of Schnyder woods
  to orientable surfaces}, HDR thesis, manuscript, 2016.

\bibitem{Meu} F. Meunier, personal communication, 2015.

\bibitem{Moh96} B. Mohar, Straight-line representations of maps on the
  torus and other flat surfaces, \emph{Discrete Mathematics} 155
  (1996) 173-181.

\bibitem{Oss94} P.~Ossona de Mendez, \emph{Orientations bipolaires},
  PhD thesis, manuscript, 1994.



\bibitem{PS06} D.~Poulalhon, G.~Schaeffer, Optimal coding and sampling
  of triangulations, \emph{Algorithmica} 46 (2006) 505-527.


\bibitem{Pro93} J. Propp, Lattice structure for orientations of
  graphs, manuscript, 1993, arXiv:math/0209005.

\bibitem{UecPHD} T. Ueckerdt, Geometric representations of graphs with low
polygonal complexity, PhD thesis, manuscript, 2011.

\bibitem{Sch89} W. Schnyder, Planar graphs and poset dimension, {\it
    Order} 5 (1989) 323-343.
  
\end{thebibliography}
\end{document}